\documentclass{article}

\usepackage{amsmath}
\usepackage{amsfonts}
\usepackage{lscape}
\usepackage{color}
\usepackage{hyperref}
\usepackage[a4paper,inner=3.0cm,outer=3.0cm,top=3.7cm,bottom=3.5cm]{geometry}

\newcommand{\citep}{\cite}
\newcommand{\citet}{\cite}
\newcommand{\email}{\tt}
\newcommand{\qed}{}
\usepackage{graphicx}

\DeclareMathOperator{\Cov}{Cov}
\newtheorem{example}{Example}
\newtheorem{proposition}{Proposition}
\newtheorem{definition}{Definition}
\newtheorem{lemma}{Lemma}
\newtheorem{proof}{Proof}
\begin{document}
%% Here are the title, author names and addresses
\title{Functional ANOVA with Multiple Distributions: Implications for the Sensitivity Analysis of Computer Experiments%\thanks{Submitted to the editors \today.}
}

\author{Emanuele Borgonovo%
\thanks{Department of Decision Sciences, Bocconi University, Via Roentgen 1, 20836, Milan, Italy (\email{emanuele.borgonovo@unibocconi.it})}
\ and
Max D. Morris%
\thanks{Department of Statistics, Iowa State University, 2438 Osborn Drive, Ames, IA 50011-1210
  USA (\email{mmorris@iastate.edu})}{}
\ and
Elmar Plischke%
\thanks{Institut f\"{u}r Endlagerforschung, Technische Universit\"{a}t Clausthal, Adolph-Roemer-Str. 2a,
38678 Clausthal-Zellerfeld, Germany (\email{elmar.plischke@tu-clausthal.de})}{}}
\date{December 19, 2017}
\maketitle
\abstract{
The functional ANOVA expansion of a multivariate mapping plays a fundamental role in statistics. The expansion is unique once a unique distribution is assigned to the covariates. Recent investigations in the environmental and climate sciences show that analysts may not be in a position to assign a unique distribution in realistic applications. We offer a systematic investigation of existence, uniqueness, orthogonality, monotonicity and ultramodularity of the functional ANOVA expansion of a multivariate mapping when a multiplicity of distributions is assigned to the covariates. In particular, we show that a multivariate mapping can be associated with a core of probability measures that guarantee uniqueness. We obtain new results for variance decomposition and dimension distribution under mixtures. Implications for the global sensitivity analysis of computer experiments are also discussed.

Keywords: Computer Experiments; Functional ANOVA; Mixtures; Global Sensitivity Analysis.
}
\section{Introduction}\label{S_intro}

The functional ANOVA expansion \cite{EfroStei81} plays a central role in the design and analysis of computer experiments \cite{Saltelli2000}. It provides the mathematical background of modern approaches to statistical
inference in computer experiments \cite{OaklOhag04}. One of the key premises
to the current use of functional ANOVA-based methods is the unique
distribution assumption: \textit{We assume to have information about the
factors' probability distribution, either joint or marginal, with or without
correlation, and that this knowledge comes from measurements, estimates,
expert opinion, physical bounds, output from simulations, analogy with
factors for similar species, and so forth} \cite[p. 704]{saltelli_jasa_2002}. With this assumption, we obtain a unique functional ANOVA expansion and, consequently, a unique set of the associated
sensitivity measures.

However, lack of data, measurement errors, or expert disagreement may prevent analysts from
assigning a unique distribution to the model inputs. Millner et al.\ report that researchers have assigned
nineteen different distributions to climate sensitivity in alternative scientific investigations of the past ten years \cite{MillDietHeal12}. Gao et al.\ show high variability in
computer experiments performed under alternative scenarios \cite{Gao2016}. Paleari and Confalonieri test robustness of sensitivity results for uncertainty in distribution using the WARM model as a case study \cite{PaleConf16}. They show that uncertainty in distribution causes an overturn of the most important variables in $22\%$ of the cases. These are not the first works dealing with the robustness of a sensitivity analysis results to the choice of the model input distributions. Early on, Chick discusses the use of two-stage distributions in simulation experiments \cite{Chic01}, Hu et al.\ studies uncertainty quantification on a well known climate model under uncertainty in distribution \cite{HuCaoHong12}.  The work of Beckman and McKay is possibly the first work discussing the stability of sensitivity analysis results for perturbations in the model input distributions \cite{Beckman1987}. As Saltelli et al.\ underline, the use of multiple distributions may controversial \cite{SaltIST15}. Nonetheless, it has become a de-facto part of several studies and is frequently adopted.

\textcolor{black}{Our purpose is to offer a systematic investigation of the impact of removing
the unique distribution assumption on the classical functional ANOVA expansion of a multivariate mapping. }\textcolor{black}{We consider two paths that emerge from current and past practices. The starting datum is that the analyst posits a set $\mathcal{M}=\left\{\mu^1_{\mathbf{X}}(\mathbf{x}),\mu^2_{\mathbf{X}}(\mathbf{x}),\dots,\mu^Q_{\mathbf{X}}(\mathbf{x})\right\}$ of plausible model
input distributions. In the first
path, the analyst evaluates the model for each distribution in $\mathcal{M}$ separately and obtains sensitivity measures for each distribution --- without-prior path, henceforth. In
the second path, the analyst assigns a prior over the the distributions in $\mathcal{M}$ --- with-prior path henceforth. For each path, we investigate six relevant notions: existence, uniqueness, orthogonality, monotonicity, ultramodularity, variance decomposition and dimension distribution. We study the implications in light of three sensitivity analysis settings: factor prioritization, trend identification and interaction quantification.} 
 
\textcolor{black}{Let us report some of the findings. In both paths, existence is ensured if all the posited measures are compatible with the functional ANOVA expansion of the input output mapping. Regarding uniqueness, in the without-prior path the analyst is dealing with as many functional ANOVA expansions as many are the cores in $\mathcal{M}$. A core is defined as a set of probability measures that lead to identical 
expansions. Thus, one has uniqueness if all the posited measures belong to the same core. In the with-prior path, one regains uniqueness: a multivariate mapping can be uniquely represented as the mixture of functional ANOVA expansions. Regarding orthogonality, in the without-prior path it is preserved. In the with-prior path, mixtures of classical functional ANOVA effects are not orthogonal with respect to the mixture of the distributions in $\mathcal{M}$. Regarding monotonicity and ultramodularity, in the without-prior path these properties are preserved by the first order effects of the classical functional ANOVA expansion under each measure in $\mathcal{M}$, as shown in previous literature \cite{Beccacece2011}. In the with-prior path, we show that they are still preserved by the mixture of first order functional ANOVA effects.}

\textcolor{black}{Regarding variance decomposition and dimension distribution, in the without-prior path, the multiplicity of variance decompositions and of dimension distributions equals the cardinality of $\mathcal{M}$. Conversely, variance decomposition and dimension distribution regain
uniqueness in the with-prior path. The variance can be decomposed as the sum of two terms,
a structural term equal to the mixture of variance decompositions and a second
generated by the variability of the model output across the
measures in $\mathcal{M}$.}
\textcolor{black}{We analyze the question of whether there are conditions under which an analyst can proceed ignoring the presence of multiple distributions. Our analysis shows that in a trend identification setting, monotonicity of the input-output mapping is a sufficient condition for the indications about trend obtained under one measure to remain the same under any other measure in $\mathcal{M}$. The same does not apply for factor prioritization and interaction quantification where the analyst needs to deal with a multiplicity of sensitivity measures, unless she(he) posits a prior. Then, the question is how to deal with such multiplicity. Formalizing the approach in \cite{Gao2016}, we propose a robust extension of the sensitivity settings of \cite{saltelli_jasa_2002} and illustrate their application through a case study.}

\textcolor{black}{The remainder of the paper is organized as follows. Section\ %
\ref{S:Literature} presents a literature review. Section\ \ref{S:Preliminary} discusses the functional ANOVA expansion in the without-prior path and introduces the notion of functional ANOVA core. Section\ \ref{S:FANOVA:Prior} addresses the with-prior path, discussing uniqueness, monotonicity, orthogonality and ultramodularity. 
Section\ \ref{S:RobustIndices} addresses variance decomposition and dimension distribution in the without-prior and with-prior paths. 
Section~\ref{S:Discussion} offers a twofold discussion, focusing first on decision-theoretical aspects and then on the link between the mixture functional ANOVA decomposition and the generalized functional ANOVA expansion \cite{Li2012,Rahm14}. 
Section\ \ref{S:Numerics} discusses numerical aspects and presents an application. 
Section \ref{S:COncl} concludes the work.}%

\section{Functional ANOVA and Related Concepts: A Review}\label{S:Literature}

The functional ANOVA expansion of a multivariate mapping originates with the
works of \cite{FishMack23} and \cite{Hoef48}, and is 
\textcolor{black}{definitively established
with the well known proofs of \cite{EfroStei81}, \cite{Sobo93} and \cite{Rabitz1999} (see \cite{Owen13SIAM} for a
detailed historical account)}. \newline
Applications of the functional ANOVA expansion are numerous. Without any
exhaustiveness claim, we recall its role in the \textit{study of quasi-Monte
Carlo integration methods, where it is applied to various notions of the
effective dimension of an integrand} \cite[p. 1]{Owen03} and in dimension
reduction in high-dimensional problems in finance \cite{Wang06}. It is
fundamental for smoothing spline ANOVA models \cite%
{Wahba78,Guo02,LinZhang06,Huang98JMA}, as well as for generalized regression
models \cite{HuangAS98,HuangetalAS99,KaufSain10Bayesian}. It is also used in conjunction with metamodelling methods, such as Gaussian metamodelling 
\cite{OaklOhag04}, polynomial chaos expansion \cite{Wiener38,CameMart47,XiuKarn02,Sudret2008}  and \textcolor{black}{polynomial dimensional
decomposition \cite{Rahman20082091}. The latter finds applications in multi-scale fracture mechanics \cite{Rahman201127}, random eigenvalue
problems \cite{RahmYada11IJUQ}, and stochastic design optimization \cite{Ren2016425}.}

Regarding the mathematical framework, one writes the input-output mapping as 
\begin{equation}
g:\mathcal{X}\rightarrow \mathbb{R},  \label{e:g:x}
\end{equation}%
where $\mathcal{X}\subseteq \mathbb{R}^{n}$ and $n$ is the number of inputs.
Under uncertainty, we denote the input probability space by $(\mathcal{X%
},\mathcal{B}(\mathcal{X}),\mu )$, where $\mu :\mathcal{B}(\mathcal{X}%
)\rightarrow \lbrack 0,1]$ represents the input probability distribution.
Uncertainty in the input causes the model output to become a function of
random variables, $G=g(\mathbf{X})$. 

\textcolor{black}{Consider now the set $Z=\{1,2,\dots ,n\}$ of the $ n $ model input indices}, and let $2^{Z}$
denote the associated power set. In the remainder, $z\in 2^{Z}$ denotes a
generic subset of indices. \textcolor{black}{Throughout the work, we assume that $g\in \mathcal{L}^{2}(%
\mathcal{X},\mathcal{B}(\mathcal{X}),\mu )$ and that $d\mu
(\mathbf{x})=\prod_{t=1}^{n}d\mu _{t}(x_{t})$, unless noted otherwise.}

\begin{proposition}
\cite{EfroStei81} \label{T_FANOVA} Under the above assumptions, $g$ can be
integrally expanded as the sum of $2^{n}$ effects
\begin{equation}
g(\mathbf{x})=\sum_{z\in 2^{Z}}g_{z}^{\mu }(\mathbf{x}_{z}),
\label{e_FANOVA}
\end{equation}%
where %\begin{equation}
%\begin{array}{c}
\begin{equation}
\begin{array}{c c c}
g_{\emptyset }^{\mu }=g_{0}^{\mu } =\mathbb{E}_{\mu }[G]=\int_{\mathcal{X}%
}g(\mathbf{x})d\mu (\mathbf{x}) & \text{,   } &
g_{z}^{\mu }(\mathbf{x}_{z})=\int_{\mathcal{X}_{\sim z}}g(\mathbf{x})d\mu (%
\mathbf{x}_{\sim z})-\sum_{v\subset z}g_{v}^{\mu }(\mathbf{x}_{v}),
\end{array}%
\end{equation}%
$d\mu (\mathbf{x}_{\sim z})=\prod_{t\not\in z}d\mu _{t}(x_{t})$ and $%
\mathcal{X}=\mathcal{X}_{z}\times \mathcal{X}_{\sim z}$, $\mathcal{X}%
_{z}\subseteq \mathbb{R}^{|z|}$, $\mathcal{X}_{\sim z}\subseteq \mathbb{R}%
^{n-|z|}$.
\end{proposition}

The function $g_{z}^{\mu }(\mathbf{x}_{z})$ is called effect function of
order $k$, where $k=|z|$ is the cardinality of the index set $z$. It refers
to the residual interaction of the inputs whose indices are in $z$. In
Proposition \ref{T_FANOVA}, $\mu $ is a product measure. The effect
functions satisfy the following conditions 
\textcolor{black}{named strong annihilating conditions in
\cite{Rahm14}}:%
\begin{equation}
\int\limits_{\mathbb{R}}g_{z}^{\mu }(\mathbf{x}_{z})d\mu _{i}(x_{i})=0\text{
for }i\in z\quad \text{and}\quad z\neq \emptyset \text{.}  \label{e:strong:annih}
\end{equation}%
These conditions imply that the effect functions have null expectation
and are orthogonal, i.e., 
\begin{equation}
\mathbb{E}_{\mu }[g_{z}^{\mu }(\mathbf{X}_{z})]=0\text{ for all }z\neq 0,%
\text{ and }\mathbb{E}_{\mu }[g_{z}^{\mu }(\mathbf{X}_{z^{\prime
}})g_{z^{\prime \prime }}^{\mu }(\mathbf{X}_{z^{\prime \prime }})]=0\text{
for all }z^{\prime \prime }\neq z^{\prime }.  \label{e:orthogonals}
\end{equation}%
In the remainder, also the conditional expectations are of interest, i.e., the
non-orthogonalized effect functions 
\begin{equation}
w_{z}^{\mu }(\mathbf{x}_{z})=\mathbb{E}_{\mathbb{\mu }}\left[ G|\mathbf{X}%
_{z}=\mathbf{x}_{z}\right] =\sum_{v\subseteq z}g_{v}^{\mu }(\mathbf{x}_{v}).
\label{e_mi1ik_non_orth}
\end{equation}%

\textcolor{black}{One then obtains the decomposition of the variance of $G$ in
$2^n-1$ terms following the steps in \cite{EfroStei81,sobol_mmce_1993}. In
particular, by subtracting $g_0^{\mu}$  from both sides, squaring, and taking the expectation we obtain}
\begin{equation}
\mathbb{E}_{\mathbb{\mu }}\left[(g(\mathbf{X})-g_0^{\mu})^{2}\right]=\mathbb{E}_{\mathbb{%
\mu }}\big[(\sum_{z\in 2^{Z},z\neq \emptyset }\!\!\!g_{z}^{\mu }(\mathbf{X}_{z}))^{2}\big].
\label{e:Variance:before:orthog}
\end{equation}%
%TCIMACRO{\TeXButton{TextcolorRedLeft}\TeXButto}}%
%BeginExpansion
%EndExpansion
{\textcolor{black}{The left-hand side of \eqref%
{e:Variance:before:orthog} is the variance of $g(\mathbf{X})$, $\mathbb{V}%
^{\mu }[G]$. Then, as a consequence of the conditions in \eqref%
{e:orthogonals}, we have}%
%TCIMACRO{\TeXButton{TextcolorRedRight}{} }%
%BeginExpansion

%EndExpansion
\begin{equation}
\text{%
%TCIMACRO{\TeXButton{TextcolorRedLeft}{\textcolor{black}{}}}%
%BeginExpansion
\textcolor{black}{}%
%EndExpansion
}\mathbb{V}^{\mu }[G]=\sum_{z\in 2^{Z},z\neq \emptyset }\int\nolimits_{%
\mathcal{X}_{z}}[g_{z}^{\mu }(\mathbf{x}_{z})]^{2}d\mu (\mathbf{x}_{z}).%
\text{%
%TCIMACRO{\TeXButton{TextcolorRedRight}{}}%
%BeginExpansion
%
%EndExpansion
}
\end{equation}
%TCIMACRO{\TeXButton{TextcolorRedLeft}{\textcolor{black}{}}}%
%BeginExpansion
\textcolor{black}{In summary, we have the following result}.%
%TCIMACRO{\TeXButton{TextcolorRedRight}{}}%
%BeginExpansion
%
%EndExpansion

\begin{proposition}
\cite{EfroStei81,sobol_mmce_1993} \label{T:dec:Variance} \textcolor{black}{The variance
of $G$ under $\mu$,\ $\mathbb{V}^{\mu }[G]$, can be written as}: 
\begin{equation}
\mathbb{V}^{\mu }[G]=\sum\limits_{z\in 2^{Z},z\neq \emptyset }V_{z}^{\mu }
\label{e:V:mu:z}
%\end{equation}%
\quad\text{where}\quad 
%\begin{equation}
V_{z}^{\mu }=\int\nolimits_{\mathcal{X}_{z}}[g_{z}^{\mu }(\mathbf{x}%
_{z})]^{2}d\mu (\mathbf{x}_{z}).  %\label{e:V:psi:z}
\end{equation}
\end{proposition}

The term $V_{z}^{\mu }$ represents the portion of $\mathbb{V}^{\mu }[G]$
caused by the interactions of the inputs with indices in $z$. The variance
decomposition in ~\eqref{e:V:mu:z} is the basis for the definition of
variance-based sensitivity indices, which are obtained by normalizing $%
V_{z}^{\mu }$ \cite{Homma1996,sobol_mmce_1993}:%
\begin{equation}
S_{z}^{\mu }=V_{z}^{\mu }/\mathbb{V}^{\mu }[G].
\end{equation}%
The quantity $S_{z}^{\mu }$ is called the variance-based sensitivity index
of group $z$ for all $z\in 2^{Z}$, $z\neq \emptyset $.

In the
formal parts of this work, we shall use the non-normalized version of the
indices $V_{z}^{\mu }$, for notational simplicity. We shall use the normalized
version $S_{z}^{\mu }$ in numerical experiments/examples. The literature has
placed particular emphasis on the first and total order sensitivity indices,
defined respectively as 
\begin{equation}
V_{i}^{\mu }=V_{\{i\}}^{\mu }   \label{S:1:mu}
%\end{equation}%
\quad\text{and}\quad 
%\begin{equation}
VT_{i}^{\mu }=\!\!\!\sum_{z:i\in z,z\neq \emptyset }V_{z}^{\mu }.
%\label{e:ST:mu}
\end{equation}%
The quantities $V_{i}^{\mu }$ and $VT_{i}^{\mu }$ are the individual and the
total contribution of $X_{i}$ to the variance of $G$. 

%TCIMACRO{\TeXButton{TextcolorRedLeft}{\textcolor{black}{}}}%
%BeginExpansion
\textcolor{black}{The estimation of variance-based sensitivity measures has been subject of
intensive studies since the late 1990's and is still an active field of research \cite{Owen12b,Owen13SIAM}. Indeed, the highest computational cost for estimation of
all variance-based indices is $C_{all}^{BF}=(2^{n}-1)N_{int}N_{out}$, where $%
N_{int}$ and $N_{out}$ are the sample sizes required for the inner and outer
loops of model evaluation associated with a brute force estimation. However,
the Extended FAST approach of \cite{Saltelli1999} allows us to obtain first
and total order indices at a cost proportional to $nN_{r}$, where 
$N_{r}$ is an appropriate number of replicates. The pick and freeze design developed in
the works of Sobol' \cite{Sobo93}, Homma and Saltelli \cite{Homma1996}, and
its amelioration in \cite{Saltelli_cpc_2002} and \cite{Saltelli2010a} allows
the estimation of all first and total order effects at a cost of $N(n+2)$
model runs, where $N$ is the basic sample size. The
random balance design, a variant of the FAST method introduced in \cite%
{Tarantola2006}, enables the estimation of first order sensitivity indices at a nominal cost of
$N$ model runs. This is the same computational cost of a so-called given data estimation. A given data estimation computes global sensitivity measures from the sample available after an uncertainty quantification. That is, one generates a sample of size $N$ for uncertainty quantification and then the \textcolor{black}{same} sample is used to estimate global sensitivity measures. Due to space
constraints, we cannot give a detailed formulation of the given data
approach and we refer the interested reader to 
\cite{StroOaki12JRSSC,PlisBorg13,BorgHazePlish15} for further details. The COSI
method introduced in \cite{Plis12EMS} is a variant of the FAST method that
permits the estimation of first order sensitivity indices from the sample generated
for an uncertainty quantification. These approaches encounter
limitations when the estimation of higher order indices is of interest. To
this purpose, a strategy in which the $N$ model runs are used to fit a
metamodel and then the metamodel is used to estimate sensitivity indices may
be more effective. Subroutines based on polynomial chaos expansion \cite%
{Sudret2008,Crestaux_ress_2009}, polynomial dimensional
decomposition \cite{Rahman2011a}, smoothing spline ANOVA models \cite{Ratto2010,ratto_cpc_2007}, sparse grid interpolation \cite{Buzz12}
are available and have found application in several disciplines. Most of
these subroutines allow the analyst to obtain estimates of higher order and
total order indices. Moreover, they allow the estimation and graphing of
first and second order effects of the functional ANOVA
expansion.}
%TCIMACRO{\TeXButton{TextcolorRedRight}{}}%
%BeginExpansion
%
%EndExpansion

%TCIMACRO{\TeXButton{TextcolorRedLeft}{\textcolor{black}{}}}%
%BeginExpansion
\textcolor{black}{}%
%EndExpansion
We conclude this review with the process of making inference in sensitivity
analysis. This process is made systematic through the concept of sensitivity
analysis setting --- see \cite{Saltelli2008} and \cite{BorgPlis15EJOR}. In a
factor prioritization, \textit{we are asked to bet on the input that, if determined (i.e., fixed to
its true value), would lead to the greatest reduction in the variance of the
model output} \cite[p. 705]{saltelli_jasa_2002}. Appropriate sensitivity
measures for this setting are variance-based first order sensitivity
measures \cite[p. 705]{saltelli_jasa_2002}. In a trend identification
setting, we are interested in determining whether an increase/decrease in
the numerical value of the inputs leads to an increase/decrease of the model
output. Appropriate sensitivity measures for this setting are the first
order effect functions of the functional ANOVA expansion \cite{Beccacece2011}. In an interaction quantification setting, we are interested in determining whether and which
interactions are significant in determining the output response to variations
in the inputs. Here, the notions of dimension distribution and of mean
effective dimension in the superimposition and truncation sense are relevant 
\cite{CaflMoro97,Owen03}. Following \cite{Owen03}, we call:\newline
a) Owen's mass function, %the probability mass function obtained from 
defined by $\Pr (T_\mu=z)=V_{z}^{\mu }/\mathbb{V}^{\mu }[G]=S_{z}^{\mu }$; 
\newline
b) dimension distribution of $g$ in the superimposition sense the
distribution of the cardinality of $T_\mu$, $|T_\mu|$, where $T_\mu$ is the random variable associated with
Owen's mass function, and \newline
c) dimension distribution of $g$ in the truncation sense, the distribution
of $\max \{j:j\in z\}$.

Owen \cite{Owen03} then defines the effective dimensions in the superimposition and truncation sense, respectively, as the
mean values of $|T|$ and of $\max \{j:j\in z\}$, i.e, as 
\begin{align}
D_{S}^{\mu }=& \sum\nolimits_{|z|>0}|z|\Pr (T_\mu=z) = \sum\nolimits_{i=1}^n VT_i^\mu\text{,} \\
D_{T}^{\mu }=& \sum\nolimits_{|z|>0}\max \{j:j\in z\}\Pr (T_\mu=z)\text{,}
\label{e:D:T:Psi}
\end{align}%
respectively. \textcolor{black}{To illustrate, a mean effective dimension in the superimposition sense equal to unity indicates the absence of interactions. Note also that the mean effective dimension is equal to the sum of total effects, as in this sum a $k^\text{th}$ order effect is counted $k$ times. Regarding interaction quantification, the higher the value of $D_{S}^{\mu }$ or $D_{T}^{\mu }$, the higher the relevance of interactions.}

\section{Existence, Multiplicity and Robustness}\label{S:Preliminary}

The discussion in Section \ref{S:Literature} shows that all the notions and
quantities related to a functional ANOVA expansion are conditional on the
distribution $\mu $. In this section, we analyze the consequences of
removing the unique distribution assumption. 
To fix ideas, suppose that the
analyst is uncertain among $Q$ possible model input distributions. The need to consider these $Q$ distributions may come from lack of data or simply by the fact that the analyst is considering a set of measures that represent a
perturbation of a reference distribution that she(he) has assigned to the inputs. In either case, the analyst is positing a
set $\mathcal{M}$ of probability measures.
The set $\mathcal{M}$ does not appear in traditional sensitivity studies of computer
experiments. $\mathcal{M}$ might be a countable (possibly infinite) or uncountable
set. The first case occurs if the decision-maker assigns a discrete number
(say $Q$) of second order probability models. The second case occurs, for
instance, in applications where the decision-maker assigns a first order
distribution depending on some parameters and then a continuous second order
distribution over the parameters. 
Each measure $\mu ^{m}$ in $\mathcal{M}$  is
associated with a potentially different functional ANOVA expansion. To illustrate, consider the next example.

\begin{example}
\label{E_non_unique}A traditional test case in sensitivity analysis is the
Ishigami test function \cite{IshiHomm90}: 
\begin{equation}
g=\sin (x_{1})\left( 1+bx_{3}^{4}\right) +a\sin ^{2}(x_{2}).  \label{e:Ishi}
\end{equation}%
The base case distribution is $\mu ^{1}:X_{1},X_{2},X_{3}\sim U[-\pi ,\pi ]$%
, i.i.d.. Then, the non-vanishing effect functions are \cite{Sudret2008}: 
\begin{equation}
g_{0}^{\mu ^{1}}=\dfrac{a}{2}; \; g_{1}^{\mu ^{1}}=\sin (x_{1})\left( 1+b%
\dfrac{\pi ^{4}}{5}\right) ; \; g_{2}^{\mu ^{1}}=a\sin ^{2}(x_{2})-\dfrac{a}{%
2}; \; g_{1,3}^{\mu ^{1}}=b\sin (x_{1})\left( x_{3}^{4}-\dfrac{\pi ^{4}}{5}%
\right) .
\end{equation}%
Due to lack of knowledge or just to test the conclusions under alternative distributions, the analyst then evaluates two alternative assignments. The second assignment is $\mu
^{2}:X_{1},X_{2},X_{3}\sim N(0,1)$, i.i.d.. Then, the non-vanishing effects of the
ANOVA decomposition are:%
\begin{equation}
\begin{split}
g_{0}^{\mu ^{2}}=\dfrac{a}{2}\left( 1-e^{-2}\right); \; g_{1}^{\mu
^{2}}=&\sin (x_{1})\left( 1+3b\right) ; \; g_{2}^{\mu ^{2}}=a\sin
^{2}(x_{2})-\dfrac{a}{2}\left( 1-e^{-2}\right) ; \\
g_{1,3}^{\mu ^{2}}&=b\sin (x_{1})\left( x_{3}^{4}-3\right) .
\end{split}
\end{equation}%
As a third assignment, she(he) sets $\mu ^{3}:X_{1},X_{2},X_{3}\sim U[0,\pi ]$, i.i.d., obtaining the
following effect functions:%
\begin{equation}
\begin{array}{ccc}
g_{0}^{\mu ^{3}}=\dfrac{a}{2}+\dfrac{2}{\pi }\left( 1+b\dfrac{\pi ^{4}}{5}%
\right) ; & g_{1}^{\mu ^{3}}=\left( \sin (x_{1})-\dfrac{2}{\pi }\right)
\left( 1+b\dfrac{\pi ^{4}}{5}\right) ; & g_{2}^{\mu ^{3}}=a\sin ^{2}(x_{2})-%
\dfrac{a}{2}; \\ 
g_{3}^{\mu ^{3}}=\dfrac{2b}{\pi }\left( x_{3}^{4}-\dfrac{\pi ^{4}}{5}\right)
; & g_{1,3}^{\mu ^{3}}=b\,\left( \sin (x_{1})-\dfrac{2}{\pi }\,\right) \,\left(
x_{3}^{4}-\dfrac{\pi ^{4}}{5}\right). &  
\end{array}%
\end{equation}%
Note that the effect function $g_{3}^{\mu ^3}(x_{3})$ is now non-null and 
that $g_2^{\mu ^3} = g_2^{\mu ^1}$.
\end{example}

In general, the set of probability measures that an analyst can posit is the
uncountable set of all distributions on measurable space $(\mathcal{X},%
\mathcal{B}(\mathcal{X}))$. However, such assignment might be either too
vast (some distributions would not reflect the analyst's state of
knowledge), or, even, incompatible with the functional ANOVA expansion of $%
g$. 
%TCIMACRO{\TeXButton{TextcolorRedLeft}{\textcolor{black}{}}}%
%BeginExpansion
\textcolor{black}{}%
%EndExpansion
In particular, an analyst's degree of belief about the model inputs is
consistent with the functional ANOVA expansion of the output only if $g$ is
measurable with respect to all the assigned distributions.%
%TCIMACRO{\TeXButton{TextcolorRedRight}{}}%
%BeginExpansion
%
%EndExpansion
} We then let 
\begin{equation}
\Psi \lbrack g]=\left\{ \mu :g\in \mathcal{L}^{2}(\mathcal{X},\mathcal{B}(%
\mathcal{X}),\mu )\wedge d\mu =\prod\nolimits_{i=1}^{n}d\mu _{i}\right\} 
\label{e:M}
\end{equation}%
denote the set of all probability measures on $(\mathcal{X},\mathcal{B}(%
\mathcal{X}))$ compatible with the functional ANOVA expansion of $g$. For
non-triviality, in the remainder, we assume that the posited set $\mathcal{M}$ is a
subset of $\Psi \lbrack g]$.

Then, let us investigate how many distinct functional ANOVA expansions
are possible for a multivariate mapping. In the next definition, consider $%
z=\{i_{1},i_{2},\dots ,i_{k}\}$ and let $\mathcal{X}_{z}=\mathcal{X}_{i_{1}}%
\mathcal{\times X}_{i_{2}}\mathcal{\times }\dots \mathcal{X}_{i_{k}}$.

\begin{definition}
\label{D:robust:PHI copy(1)}
%Given $g\in \mathcal{L}^{1}(\mathcal{X},\mathcal{%
%B}(\mathcal{X}),\mu )$, we say that 
The set $C\subseteq \Psi \lbrack g]$ is called a core of measures for the
functional ANOVA expansion of $g$ on $(\mathcal{X},\mathcal{B}(\mathcal{X}))$
if 
\begin{equation}
g_{z}^{\mu ^{\prime }}(\mathbf{x}_{z})=g_{z}^{\mu ^{\prime \prime }}(\mathbf{%
x}_{z})  \label{e:immaterial}
\end{equation}%
for all $\mu ^{\prime },\mu ^{\prime \prime }\in C$, for all $z\in 2^{Z}$
and for all $\mathbf{x}_{z}\in \mathcal{X}_{z}$ such that $d\mu ^{\prime }(%
\mathbf{x}_{z})\neq 0$, $d\mu ^{\prime \prime }(\mathbf{x}_{z})\neq 0$.
\end{definition}

Equation \eqref{e:immaterial} suggests that two probability measures in $\Psi%
[g] $ belong to the same core if they lead to the same functional ANOVA
expansion of $g$. The condition $d\mu ^{\prime }(\mathbf{x}_{z})\neq 0$, $%
d\mu ^{\prime \prime }(\mathbf{x}_{z})\neq 0$ is technical and takes into
account the situation in which the distributions differ in their support $%
\mathcal{X}$. This situation might emerge 
if $\mu ^{\prime}$ is absolutely continuous with respect
to $\mu ^{\prime \prime }$. In that case, comparing $g_{z}^{\mu ^{\prime }}(%
\mathbf{x}_{z})$ and $g_{z}^{\mu ^{\prime \prime }}(\mathbf{x}_{z})$ is
meaningful only if $\mathbf{x}_{z}$ belongs to the support of both $\mu
^{\prime }$ and $\mu ^{\prime \prime }$.

A functional ANOVA core can either contain a unique measure or a
multiplicity of measures. However, the same measure $\mu $ cannot belong
simultaneously to two cores. Then, $\Psi[g] $ is partitioned by its cores
(please refer to Appendix A for all proofs).

\begin{proposition}
\label{P:PSIunionPHI} Let $\Psi[g]$ be the set of all probability measures compatible with the functional ANOVA expansion of $g$. Let $C_{s}$ denote a
generic core. Then $\Psi[g]=\cup C_{s},$ where $C_{s}$ is a functional ANOVA
core of $\Psi[g]$ and $C_{s}\cap C_{j}=\emptyset$. Moreover, given $%
\mathcal{M}\subseteq$ $\Psi[g]$, let $C_{s}^{\mathcal{M}}=\mathcal{M}\cap C_{s}$. 
Then, $\mathcal{M}=\cup C_{s}^{\mathcal{M}}$
and $C_{s}^{\mathcal{M}}\cap C_{j}^{\mathcal{M}}=\emptyset$.
\end{proposition}

Proposition \ref{P:PSIunionPHI} allows us to characterize the multiplicity
of functional ANOVA expansions that an analyst is dealing with once $\mathcal{M}$ is
posited: $g$ possesses as many functional ANOVA representations as there are
cores in which $\mathcal{M}$ is partitioned.

%TCIMACRO{\TeXButton{TextcolorRedLeft}{\textcolor{black}{}}}%
%BeginExpansion
\textcolor{black}{The determination of cores is not straightforward. Also, one
would expect an infinity of cores if $\mathcal{M}$ is uncountable. However, the next
example illustrates a class of functions for which cores can be readily
identified.}%
%TCIMACRO{\TeXButton{TextcolorRedRight}{}}%
%BeginExpansion
%
%EndExpansion

\begin{example}
\label{P:multilin} 
%TCIMACRO{\TeXButton{TextcolorRedLeft}{\textcolor{black}{}}}%
%BeginExpansion
\textcolor{black}{Suppose that the input-output mapping can be written as 
a composite linear function
\begin{equation}
g(\mathbf{x})=\sum_{u\in 2^{Z}}\prod\limits_{i\in u}t_{i}(x_{i}).
\label{e:composite_multilinear}
\end{equation}
Then, given $M=\{\mu
,\mu ^{\prime }\}$ (inducing random vectors $X$ and $X^{\prime }$) with
propagated output random variables $G=g(X)$ and $G^{\prime }=g(X^{\prime })$
the ANOVA expansions of $G$ and $G^{\prime }$ coincide if $\mathbb{E}_{\mu
}[t_{i}(X_{i})]=\mathbb{E}_{\mu ^{\prime }}[t_{i}(X_{i}^{\prime })]$. Hence, if $C$ is any family of distributions
such that $\mathbb{E}_{\mu }[t_{i}(X_{i})]=\mathbb{E}_{\mu ^{\prime
}}[t_{i}(X_{i}^{\prime })]$ for all $\mu ,\mu ^{\prime }\in C$ then $C$ is a
core.} 
\end{example}

\textcolor{black}{The Ishigami function in Example \ref{E_non_unique} is of the form in \eqref{e:composite_multilinear} and can be rewritten as $%
g=t_{1}(x_{1})t_{3}(x_{3})+t_{2}(x_{2})$, with $t_{1}(x_{1})=\sin (x_{1})$, $%
t_{2}(x_{2})=a\sin ^{2}(x_{2})$ and $t_{3}(x_{3})=\left( 1+bx_{3}^{4}\right) 
$. Consider then the following three model input distributions: $\mu ^{1}$
as in Example \ref{E_non_unique}; $\mu ^{4}:$ $X_{1},X_{3}\sim U[-\pi ,\pi ]$%
, $X_{2}\sim U[-\frac{\pi }{2},\frac{\pi }{2}]$, and $\mu ^{5}:$ $%
X_{1},X_{2}\sim U[-\frac{\pi }{2},\frac{\pi }{2}]$, $X_{3}\sim U[-\pi ,\pi ]$%
. Then, $\mathbb{E}_{\mu ^{1}}[G_{i}]=\mathbb{E}_{\mu ^{4}}[G_{i}]=\mathbb{E}%
_{\mu ^{5}}[G_{i}]$ for $i=1,2,3$ and $x_{1},x_{2},x_{3}\in \lbrack -\frac{%
\pi }{2},\frac{\pi }{2}]$. Thus, $\mu ^{1}$, $\mu ^{4}$ and $\mu ^{5}$
belong to the same core.}
\textcolor{black}{In the special case in which $t_{i}(X_{i})=X_{i}$ for all $i=1,2,...,n$, two
distributions assigning the same expectations to the model inputs are in the
same core. The question is whether there are indeed models for which a
multilinear approximation holds. In that respect, uncertainty in
distribution is a relevant topic in reliability analysis and risk assessment
of complex technological systems \cite{Aven2015a}. As it is well known, the
mapping in probabilistic risk assessment models is multilinear as a function
of basic event probabilities. Then, for this class of problems functional
ANOVA cores are families of distributions that lead to the same expected
values of the model inputs. However, this is not the case in general. To
illustrate, the strain model $g(\mathbf{x})=x_{1}x_{3}^{x_{2}}$ in solid
mechanics does not satisfy the composite multilinearity assumption. We
therefore do not rely further on such assumption in the remainder of the
present investigation.}

\section{Multiple Distributions and a Prior}\label{S:FANOVA:Prior}

We now discuss the with-prior path. Under uncertainty in distribution, best practices recommend the use of a
two-stage sampling procedure \cite{Chic01}. In order to apply the procedure,
the analyst needs to assign a prior $P_{\mu }$ over the component measures $%
\mu ^{m}$ in $\mathcal{M}$. 
In this section, we deal with the technical aspects
that emerge after such assignment. We start with the probability spaces. Let $\mathcal{F}%
(\mathcal{M})$ denote the $\sigma$-algebra generated by all maps $\mu \mapsto \mu (A)$
for each $A\in \mathcal{B}(\mathcal{X})$ and for each $\mu \in \mathcal{M}$, giving
rise to the measurable space $(M,\mathcal{F}(\mathcal{M})) $. The corresponding
probability space is $(\mathcal{M},\mathcal{F}(\mathcal{M}),P_{\mu })$, and $P_{\mu }:\mathcal{F}%
(\mathcal{M})\rightarrow \left[ 0,1\right] $. Note that, because the algebra generated
by $\mathcal{M}$ is included in the algebra generated by $\Psi[g] $, the following are
equivalent: -- $a)$ assigning the prior directly on $(\Psi[g] ,\mathcal{F}(%
\Psi[g] ))$ or $b)$ using $(\mathcal{M},\mathcal{F}(\mathcal{M}))$ and assigning $P_{\mu }$
equal to zero on $\Psi[g] \backslash \mathcal{M}$. Hence, for notation simplicity, in
the remainder, the symbol $\mathcal{M}$ can be used without loss of generality. In the
case $\mathcal{M}$ is uncountable, we have to assume a density $dP_{\mu }(\mu )$, so
that the expectations are written as %\begin{equation}
$\mathbb{E}_{P_{\mu }}[g]=\int_{\mathcal{M}}\mathbb{E}_{P_{\mu }}[g]dP_{\mu }(\mu
)=\int_{\mathcal{M}}[\int_{\mathcal{X}}g(\mathbf{x)}\text{d}\mu (\mathbf{x})]dP_{\mu
}(\mu )$. %\end{equation}%
In the finite or countable case, $\mathcal{M}$ is of the type $\mathcal{M}=\{ \mu ^{1},\mu
^{2},\dots ,\mu ^{Q}\},$ and $P_{\mu }$ is a sequence of non-null numbers $%
p_{m}$, such that $\sum_{i=1}^{Q}p_{m}=1$ and 
%the symbol $\mathbb{E}_{P_{\mu }}[g]$ stands for 
%\begin{equation}
$\mathbb{E}_{P_{\mu} }[g]=\sum \nolimits_{m=1}^{Q}p_{m}\mathbb{E}_{\mu
^{m}}[g]=\sum \nolimits_{m=1}^{Q}p_{m}[\int_{\mathcal{X}}g(\mathbf{x)}\text{d%
}\mu ^{m}(\mathbf{x})]$. %\label{e:Eg:sum}
%\end{equation}%
For simplicity, %in our discussion we shall privilege the discrete notation
we will use this discrete notation in the remainder. %in  \eqref{e:Eg:sum}.

We start analyzing independence and uniqueness. First, one needs to observe that assigning a prior $P_{\mu }$ implies that the analyst's
uncertainty about the model inputs is represented by the mixture 
\begin{equation}
\mu _{\mathbf{X}}(\mathbf{x})=\sum_{m=1}^{Q}p_{m}\mu ^{m}(\mathbf{x}),
\label{e:mu:X:x}
\end{equation}%
where the weights are determined by $P_{\mu }=\{p_{1},p_{2},\dots ,p_{Q}\}$, 
$Q\in \mathbb{N}$, with $p_{m}>0$ and $\sum \nolimits_{m=1}^{Q}p_{m}=1$. %
Regarding independence, even if under each individual $\mu ^{m}$ in %
\eqref{e:mu:X:x} the model inputs are independent, under the mixture $\mu _{%
\mathbf{X}}$ they are not. In particular, we find: 
\begin{equation}
\Cov(X_{i},X_{j})=\sum_{m=1}^{Q}p_{m}\mathbb{E}_{\mu ^{m}}[X_{i}]\mathbb{E}%
_{\mu ^{m}}[X_{j}]-\mathbb{E}[X_{i}]\mathbb{E}[X_{j}],  \label{e:COVX1X2}
\end{equation}%
where $\mathbb{E}[X_{i}]=\int_{\mathcal{X}}x_{i}d\mu _{\mathbf{X}}(\mathbf{x}%
)$ and $\mathbb{E}_{\mu ^{m}}[X_{i}]=\int_{\mathcal{X}}x_{i}d\mu ^{m}(%
\mathbf{x})$, $i=1,2,\dots ,n$. By ~\eqref{e:COVX1X2} $\Cov(X_{i},X_{j})$ is
in general not null. However, it becomes null if $\mathbb{E}_{\mu
^{m}}[X_{i}]=\mathbb{E}_{\mu _{t}}[X_{i}]$ for all $m,t=1,2,\dots ,Q$. That
is, if there is agreement about the expected value of the model inputs, then
the model inputs remain uncorrelated, albeit, in principle, not independent.
More generally, we observe that the mixture in \eqref{e:mu:X:x} implies that
independence holds conditionally on $\mu =\mu ^{m}$. This situation
resembles de Finetti's exchangeability \cite{deFi37} in so far as
exchangeable random variables are independent conditionally on the value of
a parameter.

We now prove that, within a two-stage sampling procedure, the functional
ANOVA expansion regains uniqueness. \textcolor{black}{We consider here two alternative routes
for exploring sensitivity within a two-stage sampling procedure. A third approach is discussed in Section \ref{S:Discussion}.} The first
route consists of obtaining the functional ANOVA expansions for each
component measure in $\mathcal{M}$ separately and then taking their $P_{\mu }-$expectation.

\begin{proposition}
\label{T:genFANOVA} Given a prior $(\mathcal{M},\mathcal{F}(\mathcal{M}),P_{\mu })$ and a
measurable function $g:\mathcal{X}\rightarrow \mathbb{R}$, $g\in
\bigcap_{\mu^m \in \mathcal{M}} \mathcal{L}^2(\mathcal{X},\mathcal{B}(\mathcal{X}%
),\mu^m )$, then
\begin{equation}
g(\mathbf{x})=\sum_{z\in2^Z}g_{z}^{P_{\mu }}(\mathbf{x}_{z}),
\label{e_gen_FANOVA}
\end{equation}%
where 
\begin{equation}
g_{\emptyset }^{P_{\mu }}=\sum_{m=1}^{Q}p_{m}g_{\emptyset }^{\mu ^{m}} \text{
and } g_{z}^{P_{\mu }}(\mathbf{x}_{z})=\mathbb{E}_{P_{\mu }}[g_{z}^{\mu }(%
\mathbf{x}_{z})]=\sum_{m=1}^{Q}p_{m}g_{z}^{\mu ^{m}}(\mathbf{x}_{z}).
\label{e:gz:P}
\end{equation}
\end{proposition}

We call:
\begin{itemize}
\item[1)] the expansion $\sum_{z\in2^Z}g_{z}^{\mu_{\mathbf{X}}}(\mathbf{x}%
_{z})$ at the right hand side of \eqref{e_gen_FANOVA} %{[}{]} 
mixture functional ANOVA expansion of $g$;%\newline
\item[2)] the summands $g_{z}^{\mu _{\mathbf{X}}}(\mathbf{x}_{z})$ mixture
effect functions.
\end{itemize}

The next example illustrates Proposition \ref{T:genFANOVA}. 
\textcolor{black}{
\begin{example}{\emph{Example \ref{E_non_unique} continued.}} 
\label{E_non_unique_cont} The assigned distributions imply different supports for the model inputs. Introducing the
following indicator function $I_{[a,b]}=\begin{cases} 
1 & \text{if  $a\leq x\leq b$} \\ 
0 & \text{otherwise}\end{cases}$, we can write the three distributions in Example \ref{E_non_unique} as 
$\mu _{\mathbf{X}}^{1}=\frac{I_{[-\pi ,\pi ]}(x_{1})}{2\pi }\frac{I_{[-\pi
,\pi ]}(x_{2})}{2\pi }\frac{I_{[-\pi ,\pi ]}(x_{3})}{2\pi }$, $\mu _{\mathbf{X}}^{2}=\phi (x_{1})\phi (x_{2})\phi (x_{3})$, where $\phi (\cdot)$ is the standard Gaussian density, and $\mu _{\mathbf{X}}^{3}=\frac{I_{[0,\pi ]}(x_{1})}{\pi }\frac{I_{[0,\pi
]}(x_{2})}{\pi }\frac{I_{[0,\pi ]}(x_{3})}{\pi }$. 
Then, assigning $P_{\mu }=(\frac{1}{3},\frac{1}{3},\frac{1}{3})$, at a generic point $\mathbf{x}\in \mathbb{R}^3$, we have:
\begin{equation}
\begin{split}
g_{0}^{\mu }=&\frac{a}{2}+\frac{1}{3}\left( \dfrac{2}{\pi }\left( 1+b\frac{\pi
^{4}}{5}\right) \right) -\dfrac{a}{6}e^{-2}; \\ 
g_{1}^{\mu }(x_{1})=&\dfrac{I_{[-\pi ,\pi ]}(x_{1})}{3}\sin (x_{1})\left( 1+b\dfrac{\pi ^{4}}{5}\right) +\dfrac{\sin (x_{1})\left( 1+3b\right) }{3}  \\ & \qquad +\dfrac{I_{[0,\pi ]}(x_{1})}{3}\left( \sin (x_{1})-\frac{2}{\pi }\right)
\left( 1+b\dfrac{\pi ^{4}}{5}\right) ; \\ 
g_{2}^{\mu }(x_{2})=&\dfrac{I_{[-\pi ,\pi ]}(x_{2})}{3}(a\sin ^{2}(x_{2})-\dfrac{a}{2})+\dfrac{(a\sin ^{2}(x_{2})-a/2\left( 1-e^{-2}\right) )}{3} \\ & \qquad +\dfrac{I_{[0,\pi ]}(x_{2})}{3}(a\sin ^{2}(x_{2})-\dfrac{a}{2}); \\ 
g_{3}^{\mu }(x_{3})=&\dfrac{I_{[0,\pi ]}(x_{3})}{3}\frac{2}{3\pi }\left(
x_{3}^{4}-\frac{\pi ^{4}}{5}\right); \\ 
g_{1,3}^{\mu }(x_{1},x_{3})=& 
\dfrac{I_{[-\pi ,\pi ]}(x_{1})I_{[-\pi ,\pi ]}(x_{3})}{3}b\sin (x_{1})\left(
x_{3}^{4}-\dfrac{\pi ^{4}}{5}\right) +
\dfrac{b\sin (x_{1})\left( x_{3}^{4}-3\right) }{3} \\ & \qquad +\dfrac{I_{[0,\pi
]}(x_{1})I_{[0,\pi ]}(x_{3})}{3}b\left( \sin (x_{1})-\dfrac{2}{\pi }\right)
\left( x_{3}^{4}-\dfrac{\pi ^{4}}{5}\right).
\end{split} 
\label{e:g:Ishi:general}
\end{equation}
Thus, at the intersection of the supports, i.e., for $\mathbf{x}\in \lbrack 0,\pi ]^{3}$, we have\begin{equation}
\begin{split}
g_{0}^{\mu }=&\frac{a}{2}+\frac{1}{3}\left( \frac{2}{\pi }\left( 1+b\frac{\pi
^{4}}{5}\right) \right) -\dfrac{a}{6}e^{-2};  \\ 
g_{1}^{\mu }(x_{1})=&\left( 1+b\frac{\pi ^{4}}{5}\right) \left( \frac{2}{3}\sin (x_{1})-\frac{2}{3\pi }\right) +\sin (x_{1})\left( \frac{1}{3}-b\right)
; \\ 
g_{2}^{\mu }(x_{2})=&a\sin ^{2}(x_{2})-\frac{a}{2}\left( 1-\frac{1}{3}e^{-2}\right) ;  \qquad
g_{3}^{\mu }(x_{3})=\frac{2}{3\pi }\left( x_{3}^{4}-\frac{\pi ^{4}}{5}\right) ; \\ 
g_{1,3}^{\mu }(x_{1},x_{3})=&b\left( \left( (x_{3}^{4}-\frac{\pi ^{4}}{5}\right) \left( \frac{2}{3}\sin (x_{1})-\frac{2}{3\pi }\right) +\frac{1}{3}\sin (x_{1})\left( x_{3}^{4}-3\right) \right) .\end{split}\end{equation}
Note that the sum of the mixture effect functions equals the original
mapping at any point $\mathbf{x}\in \mathbb{R}^3$.
\end{example}}

\textcolor{black}{The second route is as follows}. The analyst wishes
to perform the functional ANOVA expansion using $\mu _{\mathbf{X}}$
as probability measure and computes the mixture effect functions from: 
\begin{equation}
g_{0}^{\mu _{\mathbf{X}}}=\mathbb{E}_{\mu _{\mathbf{X}}}[G]=\int_{\mathcal{X}%
}g(\mathbf{x})d\mu _{\mathbf{X}}, \quad g_{z}^{\mu _{\mathbf{X}}}(\mathbf{x}%
_{z})=\int_{\mathcal{X}_{\sim z}}g(\mathbf{x}_{z},\mathbf{x}_{\sim z})d\mu _{%
\mathbf{X}_{\sim z}}(\mathbf{x}_{\sim z})-\sum_{v\subset z}g_{v}^{\mu }(%
\mathbf{x}_{v}),  \label{e:GmuX}
\end{equation}%
\textcolor{black}{where 
\begin{equation}
\mu _{\mathbf{X}_{\sim z}}(\mathbf{x}_{\sim z})=\int_{\mathcal{X}_{z}}\mu _{\mathbf{X}}(\mathbf{x}) d\mathbf{x}_z
\label{e:mu:marginal}
\end{equation}
is the marginal distribution of $\mathbf{X}_{\sim z}$.} Then, the following
proposition shows that both routes lead to the same result.

\begin{proposition}
\label{T:EmuXEPmu}For a measurable function $g$ %\in \mathcal{L}^2(\mathcal{X},\mathcal{B%
%}(\mathcal{X}),\mu _{\mathbf{X}})$, 
it holds that %we have %\begin{equation}
$\mathbb{E}_{\mu _{\mathbf{X}}}[G]=\mathbb{E}_{P_{\mu
}}[G]=\sum_{m=1}^{Q}p_{m}g_{0}^{\mu ^{m}}$ % \label{e:EX:EG}
%\end{equation}%
and $%\begin{equation}
g_{z}^{\mu_{\mathbf{X}}}(\mathbf{x}_{z})=g_{z}^{P_{\mu }}(\mathbf{x}%
_{z}) % \label{e:gz:muX}
$ %\end{equation}%
so that $%\begin{equation}
g(\mathbf{x})=\sum_{v\subseteq z}g_{v}^{\mu_{\mathbf{X}}}(\mathbf{x}%
_{v})=\sum_{v\subseteq z}g_{v}^{P_{\mu }}(\mathbf{x}_{v}). 
%\label{e:g:grobust}
$ %\end{equation}
\end{proposition}

\textcolor{black}{
Propositions \ref{T:genFANOVA} and \ref{T:EmuXEPmu} suggest that we obtain the same functional ANOVA expansions by either of the following two routes: 
\begin{enumerate}
\item We get the functional ANOVA decomposition of $g$ under each of the measures in $\mathcal{M}$ and then mix the decompositions using weights determined by $P_{\mu}$; 
\item We decompose $g$ using  \eqref{e:GmuX} and \eqref{e:mu:marginal}. 
\end{enumerate}}

If we leave the sensitivity framework for a more general perspective,
Propositions \ref{T:genFANOVA} and \ref{T:EmuXEPmu} suggest a representation
theorem for a multivariate mapping. Given a set of measures $\mathcal{M}\in \Psi[g] $
and a prior $P_{\mu },$ a measurable multivariate mapping can be uniquely
projected onto $2^{n}-1$ mixture effect functions.

\textcolor{black}{
The next example illustrates the second route by means of our running example.}

\begin{example}
\textcolor{black}{For notation simplicity, let us denote the Ishigami model as a generic
three-variate mapping $g(\mathbf{x})=g(x_{1},x_{2},x_{3})$, $g:\mathbb{R}^{3}\rightarrow 
\mathbb{R}$. Also, let us denote the
three joint model input distributions in Example \ref{E_non_unique} as $\mu _{\mathbf{X}}^{m}(\mathbf{x})=\mu _{1}^{m}(x_{1})\mu _{2}^{m}(x_{1})\mu _{3}^{m}(x_{1})$, $m=1,2,3$. 
To illustrate calculations, we focus on the first order mixture effect function of $x_1$. By \eqref{e:GmuX} we write
\begin{equation}
g_{1}^{P_{\mu }}(x_{1})=\iint g(\mathbf{x})f _{2,3}(x_{2,}x_{3})dx_{2}dx_{3},
\end{equation}where, $f _{2,3}(x_{2,}x_{3})$ is the joint marginal density of $X_{2}$
and $X_{3}$, $ f _{2,3}(x_{2,}x_{3})=
\int f_{\mathbf{X}}(\mathbf{x})dx_{1}$. \textcolor{black}{Then, we have}\begin{equation}\begin{split}
g_{1}^{P_{\mu }}(x_{1})=&\iint g(\mathbf{x})f_{2,3}(x_{2,}x_{3})dx_{2}dx_{3} 
=p_{1}\iint g(\mathbf{x})f_{2}^{1}f_{3}^{1}dx_{2}dx_{3}\\  
=&p_{1}g_{1}^{f_{1}}(x_{1})+p_{2}g_{1}^{f_{2}}(x_{1})+p_{3}g_{1}^{f_{3}}(x_{1})=\sum_{m=1}^{3}p_m g_{1}^{f_{m}}(x_{1}).\label{e:gPmu}
\end{split}\end{equation}} 
%Then, equation \ref{e:gPmu}, with the definitions of the probability measures leads
%to the same result as in the second of the equations in \eqref%
%{e:g:Ishi:general}.
\end{example}

Let us now analyze the properties of the mixture effect functions, starting
with orthogonality.

\subsection{Orthogonality}

\textcolor{black}{Orthogonality is not preserved by the mixture of functional ANOVA
terms with respect to $\mu _{\mathbf{X}}(\mathbf{x})$ as reference distribution.}

\begin{example}[Example \protect\ref{E_non_unique_cont} continued]
Consider the mixture integral over the first order mixture ANOVA effect function $%
g_{3}^{\mu }(x_{3})$ from Example \ref{E_non_unique_cont}. Here, 
\begin{equation}
\mu _{\mathbf{X}}(x_{3})=\iint \mu _{\mathbf{X}}(%
\mathbf{x})dx_{2}dx_{1}=\dfrac{1}{3}\dfrac{1}{2\pi }I_{[-\pi ,\pi ]}(x_{3})+%
\dfrac{1}{3}\phi (x_{3})+\dfrac{1}{3}\dfrac{1}{\pi }I_{[0,\pi ]}(x_{3}).
\end{equation}%
Then,%
\begin{equation}
\begin{array}{c}
\int g_{3}^{\mu }(x_{3})d\mu _{\mathbf{X}}(x_{3})=\int%
\frac{2}{3\pi }\left( x_{3}^{4}-\frac{\pi ^{4}}{5}\right) \cdot \frac{1}{3}%
\left( d\mu ^{1}(x_{3})+d\mu ^{2}(x_{3})+d\mu ^{3}(x_{3})\right) \\ 
=\int \frac{2}{9\pi }\left( x_{3}^{4}-\frac{\pi ^{4}}{5}%
\right) d\mu ^{2}(x_{3})=\frac{2}{3\pi }\left( 1-\frac{\pi ^{4}}{15}\right)
\neq 0.%
\end{array}%
\end{equation}
\end{example}

\textcolor{black}{We observe that to obtain orthogonal components of the functional ANOVA expansion, one would need to resort to the generalized functional ANOVA expansion under correlations. This aspect is discussed in detail Section \ref{S:Generalized}.}

\subsection{Monotonicity}\label{S:monotonicity}

\textcolor{black}{In computer experiments, analysts are often interested in studying the
trend of the output as
the model inputs vary. Then, determining whether the output is monotonic with respect to the inputs becomes of interest.} Let us recall the definition of monotonicity for a
multivariate mapping --- see e.g.\ \cite{Marinacci2005311}.

\begin{definition}
\label{D_incr}Given $g:\mathcal{X}\rightarrow \mathbb{R}$, we say that $g$
is non-decreasing on $\mathcal{X}$ if $g$ is such that 
\begin{equation}
g(\mathbf{x}^{1})\leq g(\mathbf{x}^{2}),  \label{e_incr}
\end{equation}%
for all $\mathbf{x}^{1},\mathbf{x}^{2}\in \mathcal{X}$ with $\mathbf{x}%
^{1}\leq \mathbf{x}^{2}$.
\end{definition}

In the above definition, the inequality $\mathbf{x}^{1}\leq \mathbf{x}^{2}$
is understood component-wise. 
%, that is $\mathbf{x}^{2}=\mathbf{x}^{1}+\Delta 
%\mathbf{x}$, with $\Delta \mathbf{x}\geq \mathbf{0}$. \ 
If the strict inequality holds in \eqref{e_incr}, then one says that $g$ is
increasing. Under the unique distribution assumption, a set of results concerning
monotonicity of the effects functions are proven in \cite{Beccacece2011}. We
synthesize the main findings in the following lemma.

\begin{lemma}
\label{L_mi_mon} If $g$ is non-decreasing then:\newline
1) all non-orthogonalized effect functions $w_{z}^{\mu }(\mathbf{x}_{z})$ are non-decreasing;\newline
2) all orthogonalized first order effect functions are
non-decreasing;\newline
3) Given $\Delta \mathbf{x}\geq 0$, \textcolor{black}{introduce for 
$[\mathbf{x}, \mathbf{x}+\Delta \mathbf{x}]\subset \mathcal{X}$ the functions}
\begin{equation}
\begin{array}{c}
\Delta w_{z}^{\mu }=w_{z}^{\mu }(\mathbf{x}_{z}+\Delta \mathbf{x}%
_{z})-w_{z}^{\mu }(\mathbf{x}_{z}), \\ 
\Delta g_{z}^{\mu }=g_{z}^{\mu }(\mathbf{x}_{z}+\Delta \mathbf{x}%
_{z})-g_{z}^{\mu }(\mathbf{x}_{z}).%
\end{array}%
\end{equation}%
If the following condition holds: 
\begin{equation}
\Delta w_{z}^{\mu }\geq \sum\limits_{v\subset z}\Delta g_{v}^{\mu }
\end{equation}%
then all effect functions in \eqref{e_FANOVA} are non-decreasing.
\end{lemma}

Thus, the graphs of the first order effect functions provide visual indications about
the monotonicity of $g$. However, note that items 1 and 2 are not sufficient
conditions. Thus, we infer from the behavior of $g_{i}(x_{i})$
that if any of the first order effect functions is not monotonic then $g$ is
not monotonic. To illustrate, in Example \ref{E_non_unique} the fact that $%
g_{1}^{\mu _{1}}$ is not monotonic suffices to state that the Ishigami
function is not monotonic. Items 1 and 2 become if and only if
conditions when the model is separable, i.e., additive or multiplicative ---
see \cite{Beccacece2011} for further details. Item 3 states a sufficient
condition for all higher order effect functions to retain the original
monotonicity of $g$. % \eqref{e_qijk_gijk_gen}{]}. 
However, this condition
is rather stringent and, in general, higher order effects do not retain the
original monotonicity of $g$.

Consider now the case in which the analyst has posited $\mathcal{M}$, but without a
prior. The results in Lemma \ref{L_mi_mon} hold for any component measure $%
\mu ^{m}.$ Thus, if $g$ is monotonic, the conditional expectations of $g$
with respect to any group of model inputs as well as the first order effect
functions are monotonic. That is, the results in Lemma \ref{L_mi_mon} are
robust to the choice of the probability measure.

The next result holds in the case in which the analyst assigns a prior and shows that the mixed effect functions retain the monotonicity properties that hold under a unique distribution.

\begin{proposition}
\label{T:rob:mon} Given $(\mathcal{M},\mathcal{F}(\mathcal{M}),P_{\mu })$, if $g$ is non-decreasing then: \newline
1) the non-orthogonalized mixture effects $w_{z}^{P_{\mu }}(\mathbf{x}_{z})$
of all orders are non-decreasing;\newline
2) the orthogonalized mixture first order effects {[}$g_{i}^{P_{\mu }}(x_{i})${]} are
non-decreasing;\newline
3) Given $\Delta \mathbf{x}>0$,if for all $\mu \in \mathcal{M}$ 
\begin{equation}
\Delta w_{z}^{\mu }\geq \sum \limits_{v\subset z}\Delta g_{v }^{\mu }
\label{e_qijk_gijk_gen}
\end{equation}%
then all effects in \eqref{e_gen_FANOVA} are non-decreasing.
\end{proposition}

\subsection{Ultramodularity\label{S:Ultramodularity}}

\textcolor{black}{Ultramodularity is a generalization of scalar convexity and
is a relevant property in multivariate utility theory, game theory, and
economics \cite{Marinacci2005311,Marinacci2008642,MilgShan94}. We show in Appendix B that results similar to the ones found for monotonicity hold for ultramodularity. Then, if a function is ultramodular, it is enough to study the behavior of the effect functions under a given $\mu ^{m}$ and the qualitative insights hold for any other
component measure in $\mathcal{M}$.}

\section{Variance Decomposition}\label{S:RobustIndices}

\textcolor{black}{Regarding existence, it suffices to assume that $\mathcal{M}\subseteq \Psi \lbrack g]$}. Then,
variance-based sensitivity indices are defined for each component measure in 
$\mathcal{M}$. With regard to uniqueness, the number of variance decompositions coincides with the cardinality of\ $\mathcal{M}$.
\begin{example}[Example \protect\ref{E_non_unique} continued]
Rows three to seven in Table \ref{T:ishi} show the variance decompositions
of the Ishigami function \textcolor{black}{ with parameters $a=7.0$ and $b=0.1$} under $\mu ^{1}$, $\mu ^{2}$ and $\mu ^{3}$,
respectively. The model output variance is highest under $%
\mu ^{1}$. Also, under this measure the interaction term $V_{1,3}^{\mu
^{1}}$ is higher than under $\mu ^{1}$ and $\mu ^{2}$.
\end{example}

\begin{table}[tbp] \centering%
\caption{Variance decomposition and effective dimension results for the
Ishigami test function under the distributions of Example \ref{E_non_unique}.}%
\begin{tabular}{|c|c|r|r|r|r|p{12em}|}
\hline
& \multicolumn{5}{|c|}{Contribution to Output Variance} & Total Variance \\ 
\cline{2-7}
\raisebox{1.5ex}[-1.5ex]{Distribution} & $z$ & 1 & 2 & 3 & 1,3 & Effective
Dimension \\ \hline
& $V_{z}^{\mu ^{1}}$ & $4.35$ & $6.13$ & $0$ & $3.37$ & $\mathbb{V}^{\mu
^{1}}[G]=13.84$ \\ \cline{2-7}
\raisebox{1.5ex}[-1.5ex]{$\mu_1$} & $S_{z}^{\mu ^{1}}$ & $0.31$ & $0.44$ & $0
$ & $0.24$ & $D_{S}^{\mu ^{1}}=1.24$ \\ \hline
& $V_{z}^{\mu ^{2}}$ & $0.73$ & $5.90$ & $0$ & $0.41$ & $\mathbb{V}^{\mu
^{2}}[G]=7.05$ \\ \cline{2-7}
\raisebox{1.5ex}[-1.5ex]{$\mu_2$} & $S_{z}^{\mu ^{2}}$ & $0.10$ & $0.84$ & $0
$ & $0.06$ & $D_{S}^{\mu ^{2}}=1.06$ \\ \hline
& $V_{z}^{\mu ^{3}}$ & $0.82$ & $6.12$ & $2.73$ & $0.64$ & $\mathbb{V}^{\mu
^{3}}[G]=10.32$ \\ \cline{2-7}
\raisebox{1.5ex}[-1.5ex]{$\mu_3$} & $S_{z}^{\mu ^{3}}$ & $0.08$ & $0.59$ & $%
0.26$ & $0.062$ & $D_{S}^{\mu ^{3}}=1.05$ \\ \hline
& $B_{z}^{\mu }$ & $1.97$ & %
$6.05$  & $0.91$ & 
$1.48$ & 
\begin{minipage}{12em}$\mathbb{V}^{\mu }[G]=11.44$ \\$\mathbb{V}_{P_{\mu }}[\mathbb{E}%
_{\mu ^{m}}[G]]=1.03$ \end{minipage}\\ \cline{2-7}
\raisebox{1.5ex}[-1.5ex]{$\mu$} & $P(T_\mu=z)$ & $0.16$ & $0.49$ & $0.07$ & $0.27
$ & $D_{S}^{\mu }=1.12$ \\ \hline
\end{tabular}%
\label{T:ishi}%
\end{table}%

We consider now the case in which the analyst posits a prior $P_{\mu }$ over
the component measures in $\mathcal{M}$.  \textcolor{black}{The following holds.}

\begin{proposition}
\label{T:genVarianceDec}Given $(\mathcal{M},\mathcal{F}(\mathcal{M}),P_{\mu })$, $g\in
\bigcup_{\mu^m \in \mathcal{M}}\mathcal{L}^{2}(\mathcal{X},\mathcal{B}(\mathcal{X}%
),\mu^m )$, let $\mathbb{V}[G]$ be the overall variance of $G$ due to
uncertainty in $\mathbf{X}$ and $P_\mu $. We have 
\begin{equation}
\mathbb{V}[G]=\sum \limits_{z\in 2^{Z},z\neq \emptyset }B_{z}+\mathbb{V}%
_{P_{\mu }}[\mathbb{E}_{\mu ^{m}}[G]],  \label{e_fANOVA_mixtures}
\end{equation}
where 
\begin{equation}
B_{z}=\mathbb{E}_{P_{\mu }}[V_{z}^{\mu ^{m}}]=\sum
\limits_{m=1}^{m}p_{m}V_{z}^{\mu ^{m}}  \label{e:VPz:pm}
\end{equation}
\textcolor{black}{is the weighted average of variance-based sensitivity indices of order z.} 
\end{proposition}

Proposition \ref{T:genVarianceDec} shows that, under a mixture of
distributions, the variance of $G$ is decomposed into two summands. The
first term is the weighted average of the variance
decompositions of $G$ under the component measures in $\mathcal{M}$ 
---
structural term, henceforth. The second summand, $\mathbb{V}_{P_{\mu }}[%
\mathbb{E}_{\mu ^{m}}[G]]$, is the residual portion of the model output
variance associated with the variation of the expected value of $G$ across
the component measures \textcolor{black}{variability-of-the-mean term, henceforth}. 
\textcolor{black}{Thus, the mixture of ANOVA decompositions would explain the model output variance completely only if there is no variability in the expected value of the model output across the measures in $\mathcal{M}$.}
\begin{example}[Example \protect\ref{E_non_unique} continued]
\textcolor{black}{For the Ishigami function with the distributions assigned in Example \ref{E_non_unique}, we register $\mathbb{V}^{\mu }[G]=11.44$, $\mathbb{V}_{P_{\mu }}[%
\mathbb{E}_{\mu ^{m}}[G]]=1.03$, and $\sum \nolimits_{z\in 2^{Z},z\neq
\emptyset }B_{z}=10.41$. Thus, the structural term explains about 
$90\%$ of the model output variance. The variability-of-the-mean term accounts for the residual $10\%$ variation.}
\end{example}

\subsection{Dimension distribution}

The results of the previous section show that, posited $\mathcal{M}$, Owen's dimension distribution becomes conditional on $\mu
^{m}$. We can write $\dfrac{V_{z}^{\mu ^{m}}}{V^{\mu _{q}}[G]}=\Pr (T_{\mu^m}=z)$, 
$m\in \mathcal{M}$. Then, the number of dimension distributions equals
the cardinality of $\mathcal{M}$. In the without-prior path, the analyst might then inspect the variability of
the mean effective dimension as the measures vary in $\mathcal{M}$ to obtain an indication about how interactions vary depending on the assigned distribution. 

In the presence of a prior, the dimension distribution regains
uniqueness. In fact, by the total probability theorem, we have: 
\begin{equation}
\Pr(T_\mu=z)=\sum_{m=1}^{Q}p_{m}\Pr (T_{\mu ^{m}}=z)=\sum_{m=1}^{Q}p_{m}%
\dfrac{V_{z}^{\mu _{m}}}{V^{\mu _{m}}[G]}  \label{e:Pr:Z:generalized}
\end{equation}%
$\Pr (T_\mu=z)$ is now the mixture of the dimension distributions obtained under
each measure $\mu ^{m}$. We then have the following result.
\begin{proposition}
\label{C:D:DRobust}Under the assumptions of Proposition \ref%
{T:genVarianceDec}, the mean effective dimensions in the superimposition and
truncation sense are unique and equal to 
\begin{equation}
D_{S}=\mathbb{E}_{P_{\mu }}[D_{S}^{\mu }],\qquad  %\label{e:D:S:robust}
%\end{equation}%
\text{and}\qquad 
%\begin{equation}
D_{T}=\mathbb{E}_{P_{\mu }}[D_{T}^{\mu }].  \label{e:D:ST:robust}
\end{equation}
\end{proposition}

Equation \eqref{e:D:ST:robust}  suggests that the
mean effective dimensions in the superimposition and truncation
senses are now the mixtures of the conditional dimensions obtained under each
measure $\mu $.
\begin{example}[Example \protect\ref{E_non_unique_cont} continued]
Referring again to Table \ref{T:ishi}, \textcolor{black}{rows nine and ten display the
	unconditional dimension distribution $P(T_\mu=z)$, for the Ishigami function, and the corresponding unconditional mean effective dimension, which equals $D_S^{\mu}=1.12$. Rows four, six and eight report the
dimension distributions under each measure: They equal $D^{\mu^1}=1.24$, $D^{\mu^2}=1.06$ and $D^{\mu^3}=1.05$ under $\mu^1$, $\mu^2$ and $\mu^3$, respectively. These results indicate that there is variability in the intensity of interaction effects, which are largest under the first measure, and are minimal under the third measure. Therefore, an analyst must place attention as to what is the measure under which the sensitivity analysis is performed when inferring insights about the strength of interaction effects.}  \newline
\end{example}

\subsection{Implications for Inference: Robust Sensitivity Analysis Settings}\label{S:Settings}
\textcolor{black}{As recommended by best practice, inference using global sensitivity analysis needs to be framed within the so-called sensitivity analysis settings \cite{Saltelli2008}. These settings have been formulated under the unique distribution assumption. We then consider the impact in the formulation of sensitivity
settings consequent to the relaxation of such assumption. To fix ideas, consider the factor prioritization setting as defined in \cite[p. 705] {saltelli_jasa_2002} \emph{We are asked to bet on the factor that,
if determined (i.e., fixed to its true value), would lead to the
greatest reduction in the variance of $G$}.
Following \cite{saltelli_jasa_2002}, the appropriate sensitivity measures
for factor prioritization are the first order sensitivity indices. Here, if
the analyst posits $\mathcal{M}$ without specifying a prior, we have a multiplicity of
sensitivity indices. If robust insights are sought, then we are making
inference in the following extended setting: \emph{We are asked to bet on the
model input that, if fixed to its true value, would lead to the greatest
expected reduction in the variance of $G$ under all measures $\mu ^{m}$ 
in $\mathcal{M}$}. The operationalization is straightforward. Let us write} 
\begin{equation}
\begin{array}{ccc}
\overline{S}_{i}=\sup_{\mu ^{m}\in \mathcal{M}}S_{i}^{\mu ^{m}} & \text{ and } & 
\underline{S}_{i}=\inf_{\mu ^{m}\in \mathcal{M}}S_{i}^{\mu ^{m}}.%
\end{array}%
\label{e:SsupSinf}
\end{equation}%
The indices $\overline{S}_{i}$ and $\underline{S}_{i}$ are the extrema of the first order variance-based sensitivity indices across
the distributions in $\mathcal{M}$. Then, we need to look for the model inputs that
satisfy: 
\begin{equation}
\underline{S}_{i}>\overline{S}_{j}\text{, for all }j\neq i,\text{ }%
j=1,2,\dots ,n.  \label{e:B:robust}
\end{equation}%
That is, in the robust factor prioritization setting, a model input is the
most important, if the inferior of the values of its first order
variance-based sensitivity indices is greater than the superior of the first
order sensitivity indices of any other model input. If one or more model
inputs satisfy the search, then they are robustly the most important model
inputs based on variance reduction. The search can be repeated for the
second most important model input, etc. If the search is satisfied for all
ranks, we say that the entire ranking is robust.

\textcolor{black}{If the analyst
assigns a prior, she(he) might be considering to use the average of the sensitivity indices over the measures in $\mathcal{M}$. Such average equals $B_i$, i.e., the mixture of the first order sensitivity indices. These indices, because they are an average, have the advantage of
synthesizing the variance-based indices $V_i^{\mu^m}$. However, they relate only to the contribution of a model
input to the structural term in \eqref{e_fANOVA_mixtures}. If the mean
variability term is preponderant, inference based on the sole structural
term might not be exhaustive}.

\begin{example}
	\label{Ex:E_non_unique_cont_1} 
	\textcolor{black}{For the Ishigami function with the
		distributions assigned in Example \ref{E_non_unique}, $X_{2}$ and $X_3$ are, respectively, the most important and least important model inputs in a robust sense. In fact, the values in Table \ref{T:ishi} show that the conditions in  \eqref{e:B:robust} is satisfied. Consequently, $X_2$ and $X_3$ are also the most and least important model inputs if the average sensitivity indices are taken as sensitivity measures. This coincidence could be expected in this case, as the structural variability term accounts for $90\%$ of
		the model output variability.}
\end{example}

\textcolor{black}{Concerning trend identification, Proposition~\ref{T:rob:mon} suggests the following. If the model is monotonic, then we obtain a monotonic trend under any measure $ \mu^m \in \mathcal{M} $ for the first order effects. Then, consider that $ \mu^m $ is the measure in use. If we register a non-monotonic trend, the model is not monotonic and the indication holds for any other measure   $ \mu^{m^\prime}\in \mathcal{M} $. Conversely, if under $ \mu^m $ we register a monotonic trend of the first order effects, we cannot conclude that the model is monotonic.} 

\textcolor{black}{Concerning interaction quantification, we can define the
inferior and superior mean effective dimensions in the superimposition and
truncation senses, respectively as 
\begin{equation}
\begin{array}{ccc}
\underline{D}_{S}=\inf_{\mu }\dfrac{\sum\nolimits_{|z|>0}V_{z}^{\mu }\cdot
|z|}{\sum\nolimits_{|z|>0}V_{z}^{\mu }} & \text{ and } & \overline{D}%
_{S}=\sup_{\mu }\dfrac{\sum\nolimits_{|z|>0}V_{z}^{\mu }\cdot |z|}{%
\sum\nolimits_{|z|>0}V_{z}^{\mu }}%
\end{array}
\label{e:DS:DTinfsup}
\end{equation}%
and 
\begin{equation}
\begin{array}{ccc}
\underline{D}_{T}=\inf_{\mu }\dfrac{\sum\nolimits_{|z|>0}V_{z}^{\mu }\max
\{j:j\in z\}}{\sum\nolimits_{|z|>0}V_{z}^{\mu }} & \text{ and } & \overline{D%
}_{T}=\sup_{\mu }\dfrac{\sum\nolimits_{|z|>0}V_{z}^{\mu }\max \{j:j\in z\}}{%
\sum\nolimits_{|z|>0}V_{z}^{\mu }}.%
\end{array}
\label{e:DT:infsup}
\end{equation}
If a prior is set, we can consider the unconditional dimension distribution and the unconditional mean effective dimensions in eq.\ \eqref{e:D:ST:robust} to obtain indications on the relevance of interactions.
% and \eqref{e:D:T:robust}. 
}

\section{Discussion}\label{S:Discussion}
\subsection{Interpretation and Methodological Aspects concerning Aggregation} 
\textcolor{black}{In this section, we discuss methodological aspects concerning the use of multiple distributions, the theoretical rationale that supports our two paths, and the aggregation of experts opinion.}  
\textcolor{black}{Concerning the circumstances that motivate the use of multiple distributions, a first case is the situation in which available data do not uniquely identify a best fit from a family of distributions --- see \cite{Chic01} among others. Closely related is the case in which a best fitting family is identified, but uncertainty remains about the values of the parameters. A first example is the application of Hu, Cao and Hong \cite{HuCaoHong12} (that we use as  starting point in our case study), who assign a multivariate normal distribution with uncertain variance to the model inputs.  A second example is the
situation in which the analyst has elicited model input
distributions from more than one expert, and the experts have provided discordant opinions. A third example is illustrated in \cite{MillDietHeal12}, and is the case in which alternative scientific studies assign different distributions to a given model input. A fourth case is
discussed in the works of \cite{Beckman1987}, \cite{Badea2008} and most recently \cite{PaleConf16}, where a
robustness question is asked directly by the analyst, who wishes to explore the
stability of sensitivity analysis results for perturbations in the model
input distributions.}

\textcolor{black}{Assigning a prior is necessary for two-stage Monte
Carlo sampling \cite{Chic01}.} \textcolor{black}{The assignment of a prior is a delicate task and has been thoroughly investigated in the literature --- see the monograph \cite{BernSmit94} for a comprehensive review. \cite{ Chic01} provides an accurate methodological summary and discusses advantages and disadvantages of four methods namely, the use of a uniform prior, the use of a non-informative prior \cite{Berger2009905}, of a data-driven prior \cite{Berger1996109} and of the moment-matching method \cite{Berger1996109}.}

\textcolor{black}{Concerning the theoretical interpretation of using or not a prior, we note a recent decision theory result in \cite{CerrMacchIPNAS}. The starting point is Wald's observation that the analyst knows only \textit{that the probability
distribution of X, that is, the probability measure function $\mu $ is in the space }$\mathcal{M}$\textit{\ }\cite[p. 279]{Wald47}.
Wald's approach leads to the minimax functional as decision
criterion. The minimax philosophy corresponds to the robust settings in Section\ \ref{S:Settings}. However, \cite[p. 975]{Cerretal13Robust} shows that, by enriching Savage's axioms with
Wald's datum, i.e. positing $\mathcal{M}$, one obtains as a decision criterion a two-stage utility
functional whose form is identical to the subjective expected
utility criterion of a Bayesian decision maker who assigns a prior $P_{\mu }$ over the
possible distributions. Thus, an expected utility decision maker who has posited $\mathcal{M}$ is, indeed, a Bayesian decision maker who averages uncertainty in distribution using the prior $P_{\mu }$. This leads directly to the mixture we have discussed.}\newline
\textcolor{black}{However, some words of caution are needed about the meaning of a mixture distribution, when the distributions come from experts.} \textcolor{black}{In fact, aggregating experts' opinion is a delicate task. Several aggregation methods are available, and their applicability depends also on the assumptions at the basis of the elicitation procedure. Because of the limited space in this work, we refer to the monographs \cite{Cook91,OHagan2006}, to the review articles \cite{Ouchi2004,PaulGartOHag05}, and to the works \cite{Aspi10,Grang14PNAS,Cooke201512} that offer analyses of critical aspects. \\
From an expert aggregation perspective, the mixture in \eqref{e:mu:X:x} can be seen as a linear opinion pool. A popular alternative is the logarithmic opinion pool \cite{PaulGartOHag05}. In this case the distribution is found by aggregating opinions through
$\mu _{\mathbf{X}}^{logpool}(\mathbf{x})=k\prod\limits_{m=1}^{Q}{{{\mu }^{m}_\mathbf{X}}{{(\mathbf{x})}^{{{w}_{m}}}}}$
where $k_i$ is a normalizing constant and the weights $w_m$ are positive and sum to unity. Now, as mentioned, in this work we consider a set $\mathcal{M}$ of product measures. We can then write $
\mu _{\mathbf{X}}^{logpool}(\mathbf{x})=k\underset{i=1}{\overset{n}{\mathop \prod }}\,\prod\limits_{m=1}^{Q}{\mu _{i}^{m}{{(x_i)}^{{{w}_{m}}}}.}$
Then, if we let $\mu _{i}^{logpool}({{x}_{i}})={{k}_{i}}\prod\limits_{m=1}^{Q}{\mu _{i}^{m}{{({{x}_{i}})}^{{{w}_{m}}}}},$ we obtain
\begin{equation}
\mu _{\mathbf{X}}^{logpool}(\mathbf{x})=\dfrac{k}{\prod_i^n{k}_{i}}\underset{i=1}{\overset{n}{\mathop \prod }}\,\mu _{i}^{logpool}({{x}_{i}}).
\label{e:logpoolmarginals}
\end{equation}
Thus, if an analyst aggregates the distributions in $\mathcal{M}$ using a logarithmic pool, the multiple distribution case is absorbed back into a unique distribution case. The unique distribution is now a product of suitably defined marginals and the properties of the classical ANOVA expansion remain unaltered in this case. That is, uncertainty in distribution impacts the model inputs marginal distributions, but not the classical functional ANOVA expansion. With a similar procedure one can analyze the joint distribution resulting from other aggregation methods and can assess the impact on the functional ANOVA expansion}.

\subsection{\textcolor{black}{Relationships with the Generalized Functional ANOVA Expansion}}\label{S:Generalized} 
\textcolor{black}{We take a step back and consider the unique distribution assumption for a moment. The distribution is $\mu$ as in Section \ref{S:Literature}. As we have seen, independence plays an important role in the classical functional ANOVA expansion in \eqref{e_FANOVA}. The works \cite{Hooker2007,ChasGamb,Li2012,Rahm14}  show that we can still recover an expansion of the form
of \eqref{e_FANOVA} without imposing any assumption on $\mu_{\mathbf{X}}(\mathbf{x})$. In particular, Rahman considers the following weak annihilating conditions \cite{Rahm14}
\begin{equation}
\begin{array}{c c c}
\int_{\mathcal{X}_{z}}g_{z}(\mathbf{x}_{z})f_{z}(\mathbf{x}_{z})dx_i=0 & \mbox{} & \mbox{for } i\in z\neq \emptyset,
\end{array}
\end{equation}
where $ f_z(\mathbf{x}_z) $ is the density of $ \mathbf{x}_z $, and shows that they lead to the generalized functional ANOVA expansion 
\begin{equation}\label{e:G:FANOVA}
g(\mathbf{x})=\sum_{z\in 2^{Z}}g_{z;R}^{\mu }(\mathbf{x}_{z}).
\end{equation}
Here the subscript $R$ denotes the fact that we are dealing with a generalized component function. Also, \cite{Rahm14} shows that the generalized component functions remain hierarchically orthogonal.\footnote{Hierarchical orthogonality is the condition 
%\begin{equation}
$\mathbb{E}[g^{\mu}_{u;R}(X_u)g^{\mu}_{v;R}(X_v)]=0$
%\end{equation}
whenever $u\subset v$, $u\neq v$.}} The determination of the generalized component functions is now obtained through a system of coupled equations. To illustrate, for a three-variate function $g(\mathbf{x})=g(x_1,x_2,x_3)$, the equation of the main effect function $g_{1;R}(x_1)$ is \cite[p. 678]{Rahm14}
\begin{equation}
\begin{split}
g_{1;R}(x_1) & =\int_{\mathcal{X}_{2,3}}g(\mathbf{x})d%
\mu_{X_2,X_3}(x_2,x_3)-g_{0;R} -\int_{\mathcal{X}_{2}}g_{1,2;R}(x_1,x_2)d%
\mu(x_2)+ \\
- &\int_{\mathcal{X}_{3}}g_{1,3;R}(x_1,x_3)d\mu(x_3)-\int_{\mathcal{X}_{2,3}}g_{1,2,3;R}(x_1,x_2,x_3)d\mu_{X_2,X_3}(x_2,x_3).  \label{e:g1:G}
\end{split}
\end{equation}
The expression in \eqref{e:g1:G} involves all the
effect functions in the decomposition that contain model input $X_1$. Thus, as opposed to the independence case, the functional ANOVA terms cannot be determined recursively. \textcolor{black}{ However, a series of recent works \cite{Li2012,LiRabi17,Rahm14} present methodologies for obtaining the generalized ANOVA terms bypassing the coupling problem. In particular, the terms of the
generalized functional ANOVA expansion can be obtained following the procedure in \cite{Rahm14} by selecting a basis made of orthonormal polynomials with respect to the mixture measure $\mu _{\mathbf{X}}(\mathbf{x})$.}\\ 
\textcolor{black}{These works also extend variance-based sensitivity indices for dependent inputs. In particular, \cite{Rahm14} shows that it is possible to decompose the variance of $G$ as
\begin{equation}
\mathbb{V}[G]=\sum_{u\in2^Z\setminus  \emptyset } \mathbb{E}_{\mu
}[g_{u;R}(\mathbf{X}_{u})^{2}]+\sum_{\substack{ u,z\in2^Z\setminus  \emptyset \\ u\not\subseteq z\not\subseteq u}}\mathbb{E}_{\mu
}[g_{u;R}(\mathbf{X}_{u})g_{z;R}(\mathbf{X}_{z})], \label{e:VG:Rahman}
\end{equation}%
and correspondingly define pairs of variance-based sensitivity indices 
\begin{equation}
S_{u}^{V}={\mathbb{V}%
[G]}^{-1}{\mathbb{E}_{\mu }[g_{u;R}(\mathbf{X}_{u})^{2}]}
\end{equation}%
and 
\begin{equation}
S_{u}^{C}={\mathbb{V}[G]}^{-1}{\sum_{\substack{u,v\in2^Z\setminus  \emptyset  \\ %
u\not\subseteq v\not\subseteq u}}\mathbb{E}_{\mu }[g_{u;R}(\mathbf{X}%
_{u})g_{v;R}(\mathbf{X}_{v})]},
\end{equation}%
where the first index $S_{u}^{V}$ refers to a variance contribution, the
second to a covariance contribution $S_{u}^{C}$ generated by the presence of correlations.}
\textcolor{black}{These results can be related to our work. In the with-prior path, in fact, the analyst posits a set of measures $\mathcal{M}$ and assigns a prior $P_{\mu }$ obtaining a joint model input distribution $\mu _{\mathbf{X}}(\mathbf{x})$, which, as we have seen, is not a product measure.  We then have the following identities:%
\begin{equation}
\sum_{z\in2^Z}g_{z}^{P_{\mu }}(\mathbf{x}_{z})=\sum_{z\in2^Z}g_{u;R}(\mathbf{x}_{u}),
\label{e:fANOVA:equalitry}
\end{equation}%
and 
	\begin{equation}\label{e:Vdecompostions}
	\sum\limits_{u\in2^Z\setminus\emptyset }B_{u}+\mathbb{V}_{P_{\mu
	}}[\mathbb{E}_{\mu ^{m}}[G]]=\sum_{u\in2^Z\setminus\emptyset }%
	\mathbb{E}_{\mu }[g_{u;R}(\mathbf{X}_{u})^{2}]+\sum_{\substack{ u,z\in2^Z\setminus\emptyset  \\ u\not\subseteq z\not\subseteq u}}\mathbb{E}_{\mu
	}[g_{u;R}(\mathbf{X}_{u})g_{z;R}(\mathbf{X}_{z})]
	\end{equation}%
Equation \eqref{e:fANOVA:equalitry} relates the mixture functional ANOVA expansion in \eqref{T:genFANOVA} (left hand side) and the
generalized decomposition with respect to the mixture distribution $\mu _{\mathbf{X}}(\mathbf{x})$. Equation \eqref{e:Vdecompostions} relates the decomposition of the model output variance across the measures in $ \mathcal{M} $ (left hand side) to the decomposition over the joint mixture distribution $ \mu_{_{\mathbf{X}}}(\mathbf{x}) $ (right hand side).}

Regarding the simultaneous relaxation of the independence and uniqueness assumptions, we hint here at some possible results, which, however, need to be formally addressed in future research. \textcolor{black}{First, if we relax the independence assumption, the decision maker posits a set  $\mathcal{M}$ consisting of joint distributions, in which $\mu_X^m$ is not necessarily a product measure.  The generalized functional ANOVA expansion in \eqref{e:G:FANOVA} becomes the relevant expansion. In the without-prior path, the analyst is then dealing with a multiplicity of generalized functional ANOVA decompositions. We argue that the definition of functional ANOVA core and
Proposition \ref{P:PSIunionPHI} would still hold, under the condition that two joint measures lead to the same generalized functional ANOVA expansion, and not to the same classical ANOVA expansion. Concerning properties such as monotonicity and ultramodularity, we would need to study these properties how the generalized component functions behave with respect to these properties. For variance decomposition in the without-prior path, we would have as many pairs of indices $ S_{u}^{V} $ and $ S_{u}^{C} $ as many are the measures in $ \mathcal{M} $. In the with-prior path, we would regain uniqueness of the expansion and of the variance decomposition. However, the disclaimer holds that the investigation of this subject requires more space that can be devoted here and is, we hope, an interesting question of future research.}

\section{Numerical Implications and An Application\label{S:Numerics}}
\subsection{Estimation: Investigation of the Computational Cost}
Uncertainty quantification in the presence of a prior $P_{\mu }$ is,
conceptually, carried out following a two-stage sampling strategy \cite%
{Chic01}. First, a distribution $\mu ^{m}$ is drawn from $\mathcal{M}$ according to $%
P_{\mu }$ and, subsequently, a sample of values of $X$ is drawn from $\mu
^{m}$. Then, conditional on assuming $\mu ^{m}$ as a probability measure for
the model inputs, the cost of estimating a global sensitivity measure is the
same as under any unique measure. We recall that the brute force estimation
of all the terms of the functional ANOVA expansion requires a double loop of
model evaluations multiplied by the number of terms, leading to a cost $%
C^{\mu ^{m}}=N^{2}(2^{n}-1)$, where $N$ is the sample size and $2^{n}-1$ is the number of terms to be
estimated. Then, the overall computational cost ($C^{BF}$) becomes, in
principle $
C^{BF}=N^{P_{\mu }}C^{\mu ^{m}}=N^{P_{\mu }}N^{2}(2^{n}-1),$
where $N^{P_{\mu }}$ is the number of sampled distributions. This cost is clearly prohibitive for most computer experiments.

\textcolor{black}{We examine a few strategies to reduce computational burden in the remainder of this section.} 
\textcolor{black}{First, we can lower $C^{\mu ^{m}}$ profiting of methods developed in previous literature. To illustrate,} the design of \cite{Salt10TotalCPC}
lowers $C^{\mu ^{m}}$ to $N(n+1)$ to obtain all first and total indices \cite%
{StroOaki12JRSSC,PlisBorg13}. \textcolor{black}{This cost can be further reduced using a \emph{given data} design. As mentioned in Section \ref{S:Literature}, this estimates variance-based sensitivity measures directly on the sample of size $ N $ generated for uncertainty quantification by a single-loop Monte Carlo or quasi-Monte Carlo scheme, making the estimation cost independent of the number of model inputs}.

\textcolor{black}{A second strategy consists of lowering $N^{P_{\mu }}$}. \textcolor{black}{For instance the
re-weighting approach of \cite{Beckman1987} permits to use a single sample of $N$ model runs. The principle of the approach is similar to importance sampling, and we refer to \cite{Beckman1987} and 
\cite{Badea2008} for additional details. Thus, combining a re-weighting approach
with a given data approach has the potential of reducing the overall cost of the analysis
to $N$ model evaluations even in the presence of multiple distributions.}

\begin{example}
\label{Ex:E_non_unique_cont_1 copy(1)}To illustrate, we consider estimating
the first order sensitivity measures for the Ishigami function. Given a
sample of $(X,G)$ generated under $\mu ^{1}$ we associate a weight $w=\frac{%
f^{2}}{f^{1}}$ with each realization, where $f^{1}$ and $f^{2}$ are the
densities corresponding to $\mu ^{1}$ and $\mu ^{2}$. Mean, conditional
means and variance under $\mu ^{2}$ can then be computed as weighted local
averages or weighted sum of squares from the original sample. These weighted
versions can be used to estimate first-order variance contributions under $%
\mu ^{2}$ when a sample under $\mu ^{1}$ is available. The estimated first
order sensitivity indices at $N=10,000$ are $\widehat{S}_{1}^{\mu ^{1}}=0.33$%
, $\widehat{S}_{2}^{\mu ^{1}}=0.45$, $\widehat{S}_{3}^{\mu ^{1}}=0.00$.
Re-weighting this sample leads to the following estimates for the
variance-based sensitivity indices under $\mu ^{2}$: $\widehat{S}_{1}^{\mu
^{2}}=0.11$, $\widehat{S}_{2}^{\mu ^{2}}=0.83$, $\widehat{S}_{3}^{\mu
^{2}}=0.01$. The numerical values are close to those of Table~\ref{T:ishi}.
\end{example}

The above example refers to the Ishigami model, in which the running time is not an issue. In the case the running time is problematic, then an efficient way to reduce computational burden is represented by fitting the original model through a metamodel. The metamodel can then be used to carry out an analysis under alternative distributions. Here, some provisions need to be taken. For instance, the support assigned by the analyst may change with the measures in $\mathcal{M}$ --- see Example \ref{E_non_unique}. Then, we may wish to train the metamodel on the distribution with the largest support. To illustrate, suppose in Example \ref{E_non_unique} we use a sample coming from the third assignment where support is $[0,1]^3$ to fit the metamodel. If we use this metamodel to replace the original model and sample from the second distribution, in some instances we might be evaluating the emulator on values of the inputs falling outside the original training, with little control of the accuracy of the metamodel on these points if deviations from linearities are present. However, this is just a first aspect that appears in the presence of multiple distributions and a full investigation is outside the reach of the present work.

\subsection{Application: Sensitivity of the DICE Model with Multiple Distributions}
\textcolor{black}{A topical field characterized by scientific ambiguity is climate change \cite{MillDietHeal12}. Analysts, in fact, are frequently unable to assign a unique distribution to the inputs of integrated assessment models. Then, the question is whether we can still obtain robust insights from sensitivity analysis of integrated assessment models under uncertainty in distribution.
Our application is motivated by the results of the very recent and
influential investigation of \citet{NBERBosetti2015}, where
the developers of six of the most widely recognized integrated assessment
models perform a thorough and systematic uncertainty analysis
of the response of integrated assessment models (IAMs) in climate
change. IAMs are sophisticated computer codes that simulate complex
phenomena related to climate change evolution. \textit{One of the
key findings is that parametric uncertainty is more important than
uncertainty in model structure} \citep[p. 1]{NBERBosetti2015}. These
findings show, once again, the need of performing a rigorous sensitivity
analysis, especially if the goal is the identification of the key-drivers of uncertainty.}

\textcolor{black}{Because our purpose is illustrative, and also for granting reproducibility
of our results, we focus on William Nordhaus' \cite{Nord08}
Dynamic Integrated Climate-Economy (DICE) model and specifically on
the baseline (no controls) case in that model.\footnote{The code is the 2007 version of the model available at \url{www.econ.yale.edu/~nordhaus/homepage/DICE_delta_v8_YUP_book_short_noexclude.GMS}.} DICE is one of the
best known IAMs and has been applied and used as a benchmark
in several studies concerning uncertainty quantification in climate
change modelling. Aside from the original uncertainty analysis in
\citet{Nord08}, a variance-based sensitivity analyses of DICE is
performed by \citet{Butler2014} and \citet{Butler2014a}. More recently,
\citet{AndeBorg12} extend the analysis by estimating also the $\delta$-importance
measure. DICE has also been used as a test
case in robust optimization contexts, in studies such as \citet{MillDietHeal12,McInLempKell11,HuCaoHong12}. As Hu et al. underline \citet{HuCaoHong12}, the starting point for an uncertainty analysis of the DICE model is the investigation performed by Nordhaus himself in Chapter 7 of \cite{Nord08}.
We report the reference distributions in Table \ref{t:Nordh}.}
\begin{table}
\centering
\caption{\textcolor{black}{Distributional Assumptions in Nordhaus (2008) original Uncertainty Quantification; Table 7-1, p. 127, \cite{Nord08}.}}
\begin{tabular}{|c|c|c|c|}
\hline 
$X_i$ & Model Input Name & Mean & Std.\ Deviation \\ 
\hline 
$X_1$ & Growth in total factor productivity & 0.0092 & 0.004 \\ 
\hline 
$X_2$  & Initial sigma growth & 0.007 & 0.002  \\ 
\hline 
$X_3$  & Climate sensitivity & 3 & 1.11 \\ 
\hline 
$X_4$  & Damage function exponential factor & 0.0028& 0.0013 \\ 
\hline 
$X_5$  & Cost of backstop in 2005 & 1170 & 468 \\ 
\hline 
$X_6$  & POPASYM & 8600 & 1892 \\ 
\hline 
$X_7$  & $b_{12}$ in carbon cycle transition matrix & 0.189 & 0.017 \\ 
\hline 
$X_8$  & Cumulative fossil fuel extraction & 6000 & 1200 \\ 
\hline
\end{tabular}
\label{t:Nordh} 
\end{table}

\textcolor{black}{Nordhaus \cite{Nord08} (p. 126) remarks:  \emph{It should be emphasized that these distributions are indeed judgmental and have been estimated by the author. Other researchers would make, and other studies have made, different assessments of the values of these parameters}. Indeed, alternative distributions are used in subsequent studies. For instance, in \cite{McInKell08} and \cite{Hall2012} only four
uncertain model inputs are considered and the assigned distributions differ from the ones in the original study of \cite{Nord08}, and in \cite{HuCaoHong12} a second order distribution over the mean and variances of Table \ref{t:Nordh}.}
\textcolor{black}{In \cite[Section 4.3]{HuCaoHong12} the standard deviations of the model inputs are allowed a fifty percent decrease and increase a twenty percent increase. To illustrate a sensitivity analysis in this context, we discretize the variations in the standard deviations and let $\mathcal{M}=\{\mu^1_{\mathbf{X}}(\mathbf{x}),\mu^2_{\mathbf{X}}(\mathbf{x}),\dots,\mu^{19}_{\mathbf{X}}(\mathbf{x})$, where: 1) $\mu^1_{\mathbf{X}}(\mathbf{x})$ is Nordhaus's original distribution; 2) $\mu^2_{\mathbf{X}}(\mathbf{x}),\dots,\mu^9_{\mathbf{X}}(\mathbf{x})$ ($\mu^{10}_{\mathbf{X}}(\mathbf{x}),\dots,\mu^{17}_{\mathbf{X}}(\mathbf{x})$) are joint distributions with one of the model input variances shifted to its lower (upper) value, with the remaining fixed at the reference values of Table \ref{t:Nordh}; and 3) $\mu^{18}_{\mathbf{X}}(\mathbf{x})$ and $\mu^{19}_{\mathbf{X}}(\mathbf{x})$ are distributions with all model input variances at their lowest and highest values, respectively. For the with-prior path, we assign $P(\mu^1_{\mathbf{X}})=P(\mu^{18}_{\mathbf{X}})=P(\mu^{19}_{\mathbf{X}})=\frac{1}{5}$, and the remaining probabilities equal to $P(\mu^m_{\mathbf{X}})=\frac{1}{40}$ for $m=2,3,\dots,17$.}

To produce results while taking computation burden under control we proceed as follows. We consider the original distributions of \cite{Nord08} and generate a sample of size $N=10,000$. The calculations are performed in the General Algebraic Modeling System (GAMS), a platform for mathematical programming and optimization, in which the DICE model is implemented and evaluated. The $10,000$ evaluations take about 12 hours on a PC with 8GB RAM, dual core. Then, for the analysis of the remaining $18$ distributions, we train an emulator through Kriging to substitute the original model. The emulator fit registers an $R^2$ coefficient of $0.97$.
As a model output of interest we consider the change in atmospheric temperature in year 2100. Figure \ref{f:variancebased} displays the values of the first order sensitivity measures across the $19$ scenarios. b
\begin{figure}[hbtp]
\centering
\includegraphics[width=0.8\textwidth]{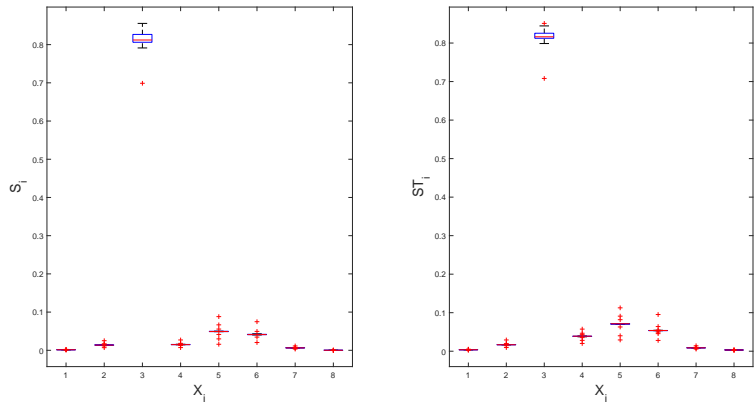}
\caption{Boxplots of the normalized first and total order variance based sensitivity measures across the 19 assigned distributions.}
\label{f:variancebased}
\end{figure}
Because we are in the without-prior path, we adopt the robust sensitivity setting of Section \ref{S:Settings}. The sensitivity measures in a robust factor prioritization setting are the first order sensitivity indices. We consider then the left graph in Figure \ref{f:variancebased}. We observe that $X_3$, climate sensitivity, is robustly the most important, because \eqref{e:B:robust} is satisfied for this model input. As for the runner up, model input $X_5$ is ranked second in $16$ out of the $19$ distribution assignments, and ranks $3^{rd}$ in two and $4^{th}$ in one. Model input $X_6$ ranks second in $3$ out of $19$ measures in $\mathcal{M}$ and ranks $3^{rd}$ otherwise. To recover robust rankings, we need to go to the least important model inputs, $X_1$, $X_7$ and $X_8$, which rank $7^{th}$, $6^{th}$ and $8^{th}$ under all measures in $\mathcal{M}$.

Concerning interaction quantification, at Nordhaus' original distribution assignment, the sum of first order indices is estimated at $0.93$, signalling a low impact of interactions. Over the additional eighteen measures, the estimate of the sum of the first order sensitivity measures ranges from a minimum of $0.92$ (last measure) to a maximum of $0.96$ (second last measure). Note that these two measures are the ones where the model input variances are at their maxima and their minima, respectively. Using a subroutine based on regression with harmonic cosine functions, we calculated the second order sensitivity indices. We register a sum of the first and second order indices close to unity, indicating that higher order interaction effects are negligible. We can then approximately compute the dimension distributions over the nineteen scenarios and the corresponding mean effective dimension in the sumperimposition sense. In the original Nordhaus
  assignment, the mean effective dimension is estimated at $1.09$. Over the additional $18$ scenarios, we register $\underline{D}_S=1.03$ and $\overline{D}_S=1.13$. Again, these values confirm the low impact of interactions.

\begin{figure}[hbtp]
\centering
\includegraphics[scale=.35]{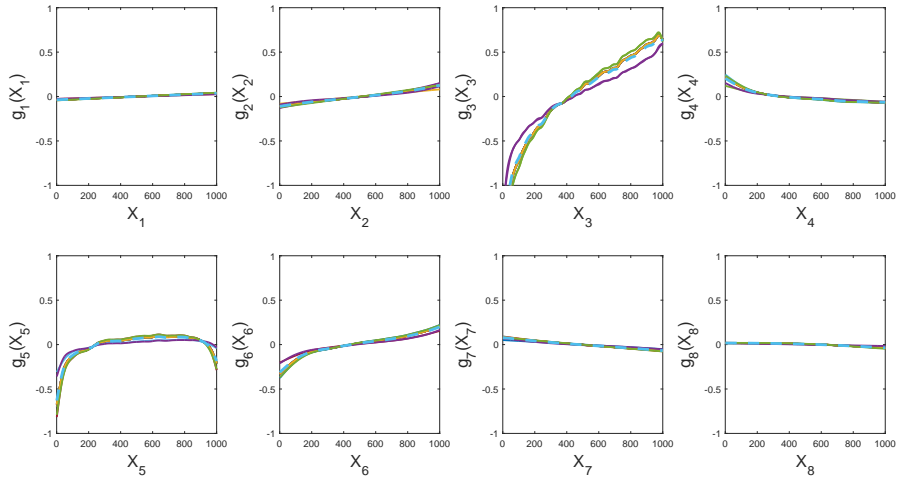}
\caption{Plots of the first order terms of the functional ANOVA expansion of the DICE model over the 19 distributions. The plots have been obtained using the SSANOVA.m code of \cite{Ratto2010}.}
\label{f:19DiceFirstOrder}
\end{figure}

\textcolor{black}{Regarding trend identification, the plots of the first order terms (Figure \ref{f:19DiceFirstOrder}) show that the model is not monotonic. In particular, the expected behavior of $G$ is decreasing in $X_4$ and $X_7$, while increasing in the remaining model inputs.  The graphs also confirm the high sensitivity of the model output on changes in $X_3$, and the low sensitivity on $X_1$, $X_7$ and $X_8$.
We observe that temperature change in $2100$ is increasing with respect to the climate sensitivity, the most important parameter. This result is in accordance with intuition, because the higher climate sensitivity the higher the expected increase in temperature.}

\textcolor{black}{In the with-prior path, one registers $\mathbb{V}[G]^{-1}\mathbb{V}_{P_{\mu}}[\mathbb{E}_{\mu^m}[G]]\approx 3.5\cdot 10^{-4}$, signaling that the fraction of $ \mathbb{V}[G] $ explained by the variation in the mean is about $1.8\%$. Because this ratio is small, we can consider the averages of the first order sensitivity indices, which is displayed in Figure \ref{f:19DiceFirstOrder}. They also indicate climate sensitivity as the most important parameter, followed by the cost of backstop in 2005. Regarding interaction quantification, the unconditional dimension distribution in \eqref{e:Pr:Z:generalized} is estimated at about $1.08$. This result signals a low impact of interactions, in agreement with the without-prior case. Regarding trend identification, the first order mixture effects are displayed in Figure \ref{f:19DiceFirstOrder} as dotted lines. The behavior of these effects is in agreement with the indications obtained in the without-prior case.}

\section{Conclusions}\label{S:COncl}

\textcolor{black}{This work has provided a first systematic study of the classical functional ANOVA expansion when the unique distribution assumption is relaxed. Such relaxation impacts properties of the expansion such as existence, uniqueness, orthogonality, ultramodularity and monotonicity. We have addressed these properties considering two main paths, depending on whether the analyst is willing to assign a prior. We have seen that, in the without-prior path, the analyst is dealing with a multiplicity of functional ANOVA expansion and, consequently, of variance-based sensitivity measures. In this context, we have introduced robust sensitivity settings. In the with-prior path, the analyst regains uniqueness and the mixture of functional ANOVA expansions equals the original mapping.}\\
\textcolor{black}{We have searched for conditions that allow the analyst to proceed as if the distribution was unique. We have discussed how insights concerning monotonicity are impacted by the assignment of a set $\mathcal{M}$ of plausible distributions. Moreover, if the analyst aggregates the distributions in $\mathcal{M}$ using a logarithmic opinion pool, then she(he) obtains a unique product distribution that allows her(him) to proceed as if the distribution was unique.
Our case study applies the findings to a well known climate model, DICE, under uncertainty in distribution, drawing from previous studies performed on the same model. Results show that when the temperature change in 2100 is the variable of interest, climate sensitivity is consistently the most important model input over the alternative distributions. Because this model input has been indicated by \cite{MillDietHeal12} among the traditional uncertainties in climate modelling, this result confirms the importance of using a robust approach when making inferences on temperature changes using the DICE model.}
\section*{Acknowledgments}

The authors wish to thank the Editor, the Associate Editor and the anonymous
reviewers for the very perceptive comments that have greatly helped us in
ameliorating the manuscript, clarifying is contributions and limitations. We
also wish to thank Steven Chick for the enlightening discussion during the
INFORMS 2012 conference in Phoenix. The authors wish also to thank Jeremy
Oakley for several constructive comments on this work. %\backmatter

\section{Appendix A: Proofs}

\begin{proof}[Proof of Proposition \protect\ref{P:PSIunionPHI}]
Consider the relation $\mu^{\prime }\overset{C}{\sim{}}\mu^{\prime \prime }$
defined by whether $\mu^{\prime }$ and $\mu^{\prime \prime }$ belong to the
same core. This is an equivalence relation. 
% Elements might have disjoint support
Passing over to the associated equivalence classes yields a partition of $%
\Psi[g]$. \qed
\end{proof}

\begin{proof}[Calculations for Example \protect\ref{P:multilin}]
Under model input independence, $g_{0}^{\mu }=\sum\limits_{u\in
2^{Z}}\prod\limits_{i\in u}\mathbb{E}[t_{i}(X_{i})]$ and $g_{0}^{\mu
^{\prime }}=\sum\limits_{u\in 2^{Z}}\prod\limits_{i\in u}\mathbb{E}%
[t_{i}(X_{i}^{\prime })]$ which coincide when all $\mathbb{E}[t_{i}(X_{i})]$
are equal to $\mathbb{E}[t_{i}(X_{i}^{\prime })]$. We consider an induction
over $|z|$. Assuming $g_{\nu }^{\mu }=g_{\nu }^{\mu ^{\prime }}$ for all $%
\nu \subset z$, $\nu \neq z$, $\mathbb{E}[t_{i}(X_{i})]=\mathbb{E}%
[t_{i}(X_{i}^{\prime })]$ for all $i$ implies $g_{z}^{\mu }=g_{z}^{\mu
^{\prime }}$, because 
\begin{equation*}
\int_{\mathcal{X}_{\sim z}}g(\mathbf{x})d\mu (\mathbf{x}_{\sim
z})=\sum_{z}\int_{\mathcal{X}_{\sim z}}\prod_{i=1}^{n}t_{i}(x_{i})d\mu (%
\mathbf{x}_{\sim z})=\sum_{z}\prod_{i\in z}t_{i}(x_{i})\cdot \prod_{i\not\in
z}\mathbb{E}[t_{i}(X_{i})]
\end{equation*}%
is equal to $\int_{\mathcal{X}_{\sim z}}g(\mathbf{x})d\mu ^{\prime }(\mathbf{%
x}_{\sim z})$.
\end{proof}

\begin{proof}[Proof of Proposition \protect\ref{T:genFANOVA}]
By Proposition \ref{T_FANOVA}, for any measure $\mu \in \Psi \lbrack g]$,
one can write the function $g$ as $g(\mathbf{x})=\sum_{z\in2^Z}g_{z}^{\mu
}(\mathbf{x}_{z})$. Then, given $(M,\mathcal{F}(M),P_{\mu })$, let us take
the expectation of both sides: 
\begin{equation}
\mathbb{E}_{P_{\mu }}[g(\mathbf{x})]=\mathbb{E}_{P_{\mu }}\left[
\sum_{z\in 2^{Z}}g_{z}^{\mu }(\mathbf{x}_{z})\right] =\sum%
\limits_{m=1}^{Q}p_{m}\sum_{z\in 2^{Z}}g_{z}^{\mu ^{m}}(\mathbf{x}_{z}).
\end{equation}%
Because $g(\mathbf{x})$ is independent of $\mu $, we obtain $\mathbb{E}%
_{P_{\mu }}[g(\mathbf{x})]=g(\mathbf{x})$. Note that this equality holds
under any measure $P_{\mu }$. For the right %left
hand side, by the linearity of the expectation operator, we have: 
\begin{equation}
\mathbb{E}_{P_{\mu }}\left[ \sum_{z\in 2^{Z}}g_{z}^{\mu^{m}}(\mathbf{x}_{z})%
\right] =\sum_{z\in 2^{Z}}\mathbb{E}_{P_{\mu }}[g_{z}^{\mu^{m}}(\mathbf{x}%
_{z})]
\end{equation}%
Then, by definition $\mathbb{E}_{P_{\mu }}[g_{z}^{\mu^{m}}(\mathbf{x}%
_{z})]=g_{z}^{P_{\mu }}(\mathbf{x}_{z})$.
\end{proof}

\begin{proof}[Proof of Proposition \protect\ref{T:EmuXEPmu}]
We start with the expected value of $G$. We have 
\begin{equation}
\mathbb{E}_{\mu_{\mathbf{X}}}[G]=\int gd\mu_{_{\mathbf{X}}}=\int
g\sum_{m=1}^{Q}p_{m}d\mu^{m}=\sum_{m=1}^{Q}p_{m}\int
gd\mu^{m}=\sum_{m=1}^{Q}p_{m}g_{\emptyset}^{\mu^{m}}
\end{equation}
The equality follows by the first equality in  \eqref{e:gz:P}. Concerning
the generic effect function, it suffices to prove the identify for conditional
expectations. We have 
\begin{multline*}
w_{z}^{\mu_{_{\mathbf{X}}}}(\mathbf{x}_{z})=\int_{\mathcal{X}}(g(\mathbf{x}%
_{z},\mathbf{x}_{\sim z}))d\mu_{\mathbf{X}}(\mathbf{x}_{\sim z})=\int g(%
\mathbf{x}_{z},\mathbf{x}_{\sim z}) \sum_{m=1}^{Q}p_{m}d\mu^{m}(\mathbf{x}%
_{\sim z})= \\
\sum_{m=1}^{Q}p_{m}\int g(\mathbf{x}_{z},\mathbf{x}_{\sim z})d\mu^{m}(%
\mathbf{x}_{\sim z})= \sum_{m=1}^{Q}p_{m}w_{z}^{\mu^{m}}(\mathbf{x}_{z})=%
\mathbb{E}_{P_{\mu}}[w_{z}^{\mu^{m}}(\mathbf{x}_{z})]=w_{z}^{P_{\mu}}(%
\mathbf{x}_{z})
\end{multline*}
\end{proof}

\begin{proof}[Proof of Proposition \protect\ref{T:rob:mon}]
Item 1. By Lemma \ref{L_mi_mon}, if $g$ is non-decreasing, then $w_{z}^{\mu
}(\mathbf{x}_{z})$ is non-decreasing for any measure $\mu $. %
Equation \eqref{e_mi1ik_non_orth} shows that a generalized non-orthogonalized effect function is
the linear combination of $w_{z}^{\mu }(\mathbf{x}_{z})$ with positive
weights. Therefore $\int w_{z}^{\mu }(\mathbf{x}_{z})dP(\mu )$ or $%
\sum_{m=1}^{Q}p_{m}w_{z}^{\mu^m}(\mathbf{x}_{z})$ is, then, non-decreasing in 
$\mathbf{x}_{z}$. $w_{z}^{P_{\mu }}(\mathbf{x}_{z})$ is, then,
non-decreasing. %\newline
Item 2 follows from the 
%. By Lemma \ref{L_mi_mon}, if $g$ is non-decreasing, then any $%
%g_{i}^{P_{\mu }}(x_{i})$ is non decreasing in $X_{i}$. Hence, applying the
same reasoning as for Item 1. %, one obtains the thesis. 
\newline
Item 3 is proven as follows. By assumption, because $p_{m}\geq 0$, (or $%
dP(\mu )\geq 0$) for a generic $\mu^{m}$ it is $
p_{m}\Delta w_{z}^{\mu^{m}}\geq p_{m}\!\sum \limits_{v\subset z}\!\!\Delta
g_{v}^{\mu^{m}},$
\textcolor{black}{where $v\subset z$ denotes true subsets, the case $v=z$ is excluded.}
Then, the following is true: 
\begin{equation}
\sum_{m=1}^{Q}p_{m}\Delta w_{z}^{\mu^{m}}\geq \sum_{m=1}^{Q}\sum
\limits_{v\subset z}p_{m}\Delta g_{v}^{\mu^{m}}=\sum
\limits_{v\subset z}\sum_{m=1}^{Q}p_{m}\Delta g_{v}^{\mu^{m}}=\sum
\limits_{v\subset z}\Delta g_{v}^{P_{\mu }}
\end{equation}%
Noting that the left-hand side is $\Delta w_{z}^{P_{\mu }}$, we have $\Delta
w_{z}^{P_{\mu }}\geq \sum \limits_{v\subset z}\Delta g_{v}^{P_{\mu }}$.
\end{proof}
\begin{proof}[Proof of Proposition \protect\ref{T:genVarianceDec}]
The following holds by the law of total variance: 
\begin{equation}
\mathbb{V}[G]=\mathbb{E}_{P_{\mu }}[\mathbb{V}\{G|\mu ^{m}\}]+\mathbb{V}%
_{P_{\mu }}\{ \mathbb{E}[G|\mu ^{m}]\}
\end{equation}%
Under the generic input distribution $\mu ^{m}$ \eqref{e:V:mu:z} applies, so
that $\mathbb{V}_{\mu ^{m}}[G]=\sum \limits_{z\in 2^{Z},z\neq \emptyset
}V_{z}^{\mu ^{m}}$, from which we obtain 
\begin{equation}
\begin{array}{c}
\mathbb{V}[G]=\mathbb{E}_{P_{\mu }}[\sum \limits_{z\in 2^{Z},z\neq \emptyset
}V_{z}^{\mu ^{m}}]+\mathbb{V}_{P_{\mu }}\{ \mathbb{E}[G|\mu ^{m}]\}= 
\sum \limits_{z\in 2^{Z},z\neq \emptyset }B_{z}+\mathbb{V}_{P_{\mu }}\{%
\mathbb{E}[G|\mu ^{m}]\}%
\end{array}%
\end{equation}%
which holds by the linearity of the summation and expectation operators.
\end{proof}

\begin{proof}[Proof of Corollary \protect\ref{C:D:DRobust}]
We start with the first equality in \eqref{e:D:ST:robust}. By definition we have 
\begin{multline*}
D_{S}=\sum\nolimits_{|z|>0}|z|\Pr (T_\mu=z)=\sum\nolimits_{|z|>0}|z|\mathbb{E}%
_{P_{\mu }}[\Pr (T_\mu=z|\mu =\mu ^{\ast })]= \\
\mathbb{E}_{P_{\mu }}\left[ \sum\nolimits_{|z|>0}|z|\Pr (T_\mu=z|\mu =\mu ^{\ast
})\right] =\mathbb{E}_{P_{\mu }}[D_{S}^{\mu }].
\end{multline*}%
For the second equality in \eqref{e:D:ST:robust}, we proceed in a similar way, obtaining: 
\begin{equation*}
\begin{array}{c}
D_{T}=\sum\nolimits_{|z|>0}\max \{j:j\in z\}\Pr
(T_\mu=z)=\sum\nolimits_{|z|>0}\max \{j:j\in z\}\mathbb{E}_{P_{\mu }}[\Pr
(T_\mu=z|\mu =\mu ^{\ast })]= \\ 
\mathbb{E}_{P_{\mu }}\left[ \sum\nolimits_{|z|>0}\max \{j:j\in z\}\Pr
(T_\mu=z|\mu =\mu ^{\ast })]=\mathbb{E}_{P_{\mu }}[D_{T}^{\mu }\right] .%
\end{array}%
\end{equation*}
\end{proof}

\section{\textcolor{black}{Appendix B: Mixture Functional ANOVA and Ultramodularity}}
\textcolor{black}{As anticipated in the main text, we discuss here the relationship between the mixture of functional ANOVA expansion and ultramodularity.} 
\begin{definition} \cite{Marinacci2005311} 
\label{D_UM_2}$g:\mathcal{X}\rightarrow \mathbb{R}$ is ultramodular, if 
\begin{equation}
g(\mathbf{x}^{1}+\Delta \mathbf{x})-g(\mathbf{x}^{1})\leq g(\mathbf{x}%
^{2}+\Delta \mathbf{x})-g(\mathbf{x}^{2})  \label{e_ultra1}
\end{equation}%
for all $\mathbf{x}^{1},\mathbf{x}^{2}\in \mathcal{X}$, $\Delta \mathbf{x}%
\geq 0$ with $\mathbf{x}^{1}\leq \mathbf{x}^{2}$ and $\mathbf{x}^{1}+\Delta 
\mathbf{x}^{1},\mathbf{x}^{2}+\Delta \mathbf{x}\in \mathcal{X}$. If the
inequality in ~\eqref{e_ultra1} is reversed, we say that $g(\cdot )$ is
neg-ultramodular.
\end{definition}

Relevant properties of ultramodular functions are discussed in \cite%
{Marinacci2005311} and \cite{Marinacci2008642}. In \cite{Beccacece2011}, the
question of whether ultramodularity is preserved within a functional ANOVA
expansion is addressed. We summarize the main results in the following lemma.
\begin{lemma}
\label{L_mi_ultram} If $g$
is ultramodular, then\newline
1) all non-orthogonalized effects $w_{z}^{\mu }(\mathbf{x}_{z})$ in its
functional ANOVA expansion are ultramodular\newline
2) all (orthogonalized) first order effects in its ANOVA expansion are
ultramodular\newline
3) Given $\Delta \mathbf{x}\geq 0$, if 
\begin{equation}
\Delta w_{z}^{\mu }(\mathbf{y})-\Delta w_{z}^{\mu }(\mathbf{x})\geq \sum
\limits_{v\subset z}\Delta g_{v}^{\mu }(\mathbf{y})-\sum \limits_{v\subset z}\Delta g_{v}^{\mu }(\mathbf{x})  \label{e_Dq_ultram}
\end{equation}%
where 
\begin{equation}
%\begin{array}{l}
\begin{array}{ccc}
\Delta w_{z}^{\mu }(\mathbf{y})=w_{z}^{\mu }(\mathbf{y}_{z}+\Delta \mathbf{x}%
_{z})-w_{z}^{\mu }(\mathbf{y}_{z}) & \text{ and }& \Delta w_{z}^{\mu }(\mathbf{x}%
)=w_{z}^{\mu }(\mathbf{x}_{z}+\Delta \mathbf{x}_{z})-w_{z}^{\mu }(\mathbf{x}%
_{z})%
%\end{array}
\\ 
%\begin{array}{ccc}
\Delta g_{z}^{\mu }(\mathbf{y})=g_{z}^{\mu }(\mathbf{y}_{z}+\Delta \mathbf{x}%
_{z})-g_{z}^{\mu }(\mathbf{y}_{z}) &\text{ and }& \Delta g_{z}^{\mu }(\mathbf{x}%
)=g_{z}^{\mu }(\mathbf{x}_{z}+\Delta \mathbf{x}_{z})-g_{z}^{\mu }(\mathbf{x}%
_{z})%
\end{array}%
%\end{array}
\label{e_DQ2_ultram}
\end{equation}%
then all effects in the functional ANOVA expansion of $g$ under measure $\mu$ are ultramodular.
\end{lemma}
The next result states that the mixed functional ANOVA expansion of $g$ preserves results obtained with a unique input distribution when
ultramodularity is concerned.
\begin{proposition}
\label{T:ultram}
Given a prior $(\mathcal{M},%
\mathcal{F}(\mathcal{M}),P_{\mu })$ and %a function
%$g:\mathcal{X}\rightarrow \mathbb{R}$, 
$g\in \bigcap_{\mu^m \in \mathcal{M}} 
\mathcal{L}^1(\mathcal{X},\mathcal{B}(\mathcal{X}),\mu^m )$.
%
%Given $g\in \mathcal{L}^{1}(\mathcal{X},\mathcal{B}(\mathcal{%
%X}),\mu )$. % with $\mu \in \Psi[g] $, %$(\Psi ,\mathcal{F}(\Psi ),M)$, 
If $g$ is
ultramodular on $\mathcal{X}$ then the following holds for its mixture
functional ANOVA expansion:\newline
1) the non-orthogonalized effects of any order, $w_{z}^{P_{\mu }}(\mathbf{x}%
_{z})$, are ultramodular;\newline
2) all orthogonalized first order effects are ultramodular;\newline
3) given $\Delta \mathbf{x}>0$, if, for all $\mu\in \mathcal{M}$ eqs. \eqref{e_Dq_ultram}
and \eqref{e_DQ2_ultram} hold, then all \textcolor{black}{effect functions} in  \eqref{e_gen_FANOVA}
are ultramodular.
\end{proposition}
\begin{proof}[Proof of Proposition \protect\ref{T:ultram}]
Item 1. By item 1 of Lemma \ref{L_mi_ultram}, the ultramodularity of $g$
ensures that $w_{z}^{\mu }(\mathbf{x}_{z})$ is ultramodular for any given
probability measure $\mu $. Then, $w_{z}^{P_{\mu }}(\mathbf{x}_{z})$ is the
convex combination of ultramodular functions, which is ultramodular by
Proposition 4.1 in \cite{Marinacci2005311}, p. 317.\newline
Item 2. By item 2 of Lemma \ref{L_mi_ultram}, if $g$ is ultramodular, then
any $g_{i}^{\mu }(x_{i})$ is ultramodular given $\mu $. Hence, because $%
g_{i}^{P_{\mu }}(x_{i})$ is a convex combination of ultramodular functions,
it is ultramodular as well by Proposition 4.1 in \cite{Marinacci2005311}, p.
317.\newline
Item 3 is proven as follows. By the assumptions of item 3, for a generic $%
\mu $ it is: 
\begin{equation}
\Delta w_{z}^{\mu }(\mathbf{y})-\Delta w_{z}^{\mu }(\mathbf{x})\geq \sum
\limits_{v\subset z}\Delta g_{v}^{\mu }(\mathbf{y})-\sum
\limits_{v\subset z}\Delta g_{v}^{\mu }(\mathbf{x})
\end{equation}%
Taking the expectation of both sides, by the linearity of the involved
operators, 
\begin{equation}
\mathbb{E}_{P_{\mu }}\left[\Delta w_{z}^{\mu }(\mathbf{y})-\Delta w_{z}^{\mu }(%
\mathbf{x})\right]\geq \mathbb{E}_{P_{\mu }}\big[\sum \limits_{v \subset z}\Delta
g_{v}^{\mu }(\mathbf{y})-\sum \limits_{v \subset z}\Delta g_{v}^{\mu }(\mathbf{%
x})\big]
\end{equation}%
we obtain 
\begin{equation}
\Delta w_{z}^{P_{\mu }}(\mathbf{y})-\Delta w_{z}^{P_{\mu }}(\mathbf{x})\geq
\sum \limits_{v\subset z}\Delta g_{z}^{P_{\mu }}(\mathbf{y})-\sum
\limits_{v\subset z}\Delta g_{z}^{P_{\mu }}(\mathbf{x})
\end{equation}
\end{proof}

%C:\Users\BorgonovoE\Dropbox\Refall
%\bibliographystyle{siamplain}
%\bibliography{C:/Users/BorgonovoE/Dropbox/Refall/library}

\begin{thebibliography}{10}
	
	\bibitem{AndeBorg12}
	{\sc B.~Anderson, E.~Borgonovo, M.~Galeotti, and R.~Roson}, {\em {Uncertainty
			in Climate Change Modelling: Can Global Sensitivity Analysis be of Help?}},
	Risk Analysis, 34 (2014), pp.~271--293.
	
	\bibitem{Aspi10}
	{\sc W.~Aspinall}, {\em {A route to more tractable expert advice}}, Nature, 463
	(2010), pp.~294--295.
	
	\bibitem{Aven2015a}
	{\sc T.~Aven}, {\em {Implications of black swans to the foundations and
			practice of risk assessment and management}}, Reliability Engineering {\&}
	System Safety, 134 (2015), pp.~83--91.
	
	\bibitem{Badea2008}
	{\sc A.~Badea and R.~Bolado}, {\em {Milestone M.2.1.D.4: Review of Sensitivity
			Analysis Methods and Experience}}, tech. report, PAMINA Project, Sixth
	Framework Programme, European Commission, 2008.
	
	\bibitem{Beccacece2011}
	{\sc F.~Beccacece and E.~Borgonovo}, {\em {Functional ANOVA, ultramodularity
			and monotonicity: Applications in multiattribute utility theory}}, European
	Journal of Operational Research, 210 (2011), pp.~326--335.
	
	\bibitem{Beckman1987}
	{\sc R.~J. Beckman and M.~D. McKay}, {\em {Monte Carlo estimation under
			different distributions using the same simulation}}, Technometrics, 29 (2)
	(1987), pp.~153--160.
	
	\bibitem{Berger2009905}
	{\sc J.~O. Berger, J.~M. Bernardo, and D.~Sun}, {\em {The formal definition of
			reference priors}}, Annals of Statistics, 37 (2009), pp.~905--938.
	
	\bibitem{Berger1996109}
	{\sc J.~O. Berger and L.~R. Pericchi}, {\em {The intrinsic bayes factor for
			model selection and prediction}}, Journal of the American Statistical
	Association, 91 (1996), pp.~109--122.
	
	\bibitem{BernSmit94}
	{\sc J.~Bernardo and A.~Smith}, {\em {Bayesian Theory}}, Wiley{\&}Sons, New
	York, NY, USA, second edi~ed., 2000.
	
	\bibitem{BorgHazePlish15}
	{\sc E.~Borgonovo, G.~Hazen, and E.~Plischke}, {\em {A Common Rationale for
			Global Sensitivity Measures and their Estimation}}, Risk Analysis, 36 (2016),
	pp.~1871--1895.
	
	\bibitem{BorgPlis15EJOR}
	{\sc E.~Borgonovo and E.~Plischke}, {\em {Sensitivity Analysis: A Review of
			Recent Advances}}, European Journal of Operational Research, 3 (2016),
	pp.~869--887.
	
	\bibitem{Butler2014a}
	{\sc M.~Butler, P.~Reed, K.~Fisher-Vanden, K.~Keller, and T.~Wagener}, {\em
		{Identifying parametric controls and dependencies in integrated assessment
			models using global sensitivity analysis}}, Environmental Modelling {\&}
	Software, 59 (2014), pp.~10--29.
	
	\bibitem{Butler2014}
	{\sc M.~Butler, P.~Reed, K.~Fisher-Vanden, K.~Keller, and T.~Wagener}, {\em
		{Inaction and climate stabilization uncertainties lead to severe economic
			risks}}, Climatic Change, 127 (2014), pp.~463--474.
	
	\bibitem{Buzz12}
	{\sc G.~T. Buzzard}, {\em {Global Sensitivity Analysis using Sparse Grid
			Interpolation and Polynomial Chaos}}, Reliability Engineering {\&} System
	Safety, 107 (2012), pp.~82--89.
	
	\bibitem{CaflMoro97}
	{\sc R.~E. Caflisch, W.~Morokoff, and A.~B. Owen}, {\em {Valuation of mortgage
			backed securities using Brownian bridges to reduce effective dimension}},
	Journal of Computational Finance, 1 (1997), pp.~27--46.
	
	\bibitem{CameMart47}
	{\sc R.~Cameron and W.~Martin}, {\em {The Orthogonal Development of Non-Linear
			Functionals in Series of Fourier-Hermite Functionals}}, Annals of
	Mathematics, 48 (1947), pp.~385--392.
	
	\bibitem{Cerretal13Robust}
	{\sc S.~Cerreia-Vioglio, F.~Maccheroni, M.~Marinacci, and L.~Montrucchio}, {\em
		{Classical Subjective Expected Utility}}, Proceedings of the National Academy of Sciences of the United States, 110 (2013),
	pp.~6754-6759.
	
	\bibitem{ChasGamb}
	{\sc G.~Chastaing, F.~Gamboa, and C.~Prieur}, {\em {Generalized Hoeffding-Sobol
			Decomposition for Dependent Variables:Application to Sensitivity Analysis}},
	Electronic Journal of Satistics, 6 (2012), pp.~2420--2448.
	
	\bibitem{Chic01}
	{\sc S.~Chick}, {\em {Input Distribution Selection for Simulation Experiments:
			Accounting for Input Uncertainty}}, Operations Research, 49 (2001),
	pp.~744--758.
	
	\bibitem{Cook91}
	{\sc R.~M. Cooke}, {\em {Experts in Uncertainty: Opinion and Subjective
			Probability in Science}}, Oxford Univ. Press, 1991.
	
	\bibitem{Cooke201512}
	{\sc R.~M. Cooke}, {\em {The Aggregation of Expert Judgment: Do Good Things
			Come to Those Who Weight?}}, Risk Analysis, 35 (2015), pp.~12--15.
	
	\bibitem{Crestaux_ress_2009}
	{\sc T.~Crestaux, O.~{Le Maitre}, and J.-M. Martinez}, {\em {Polynomial chaos
			expansion for sensitivity analysis}}, Reliability Engineering {\&} System
	Safety, 94 (2009), pp.~1161--1172.
	
	\bibitem{deFi37}
	{\sc B.~de~Finetti}, {\em {La prevision: ses lois logiques, ses sources
			subjectives. Translated to English by H.E. Kyburg and reprinted in Kyburg and
			Smokler (1964)}}, Annales de l' Istitute Henri Poincar{\'{e}}, 7 (1937),
	pp.~1--68.
	
	\bibitem{EfroStei81}
	{\sc B.~Efron and C.~Stein}, {\em {The Jackknife Estimate of Variance}}, The
	Annals of Statistics, 9 (1981), pp.~586--596.
	
	\bibitem{FishMack23}
	{\sc R.~A. Fisher and W.~A. Mackenzie}, {\em {The manurial response of
			different potato varieties}}, Journal of Agricultural Science, XIII (1923),
	pp.~311--320.
	
	\bibitem{Gao2016}
	{\sc L.~Gao, B.~A. Bryan, M.~Nolan, J.~D. Connor, X.~Song, and G.~Zhao}, {\em
		{Robust global sensitivity analysis under deep uncertainty via scenario
			analysis}}, Environmental Modelling {\&} Software, 76 (2016), pp.~154--166.
	
	\bibitem{PaulGartOHag05}
	{\sc P.~Garthwaite, J.~Kadane, and A.~O'Hagan}, {\em {Statistical Methods for
			Eliciting Probability Distributions}}, Journal of the American Statistical
	Association, 100 (2005), pp.~680--700.
	
	\bibitem{NBERBosetti2015}
	{\sc K.~Gillingham, W.~Nordhaus, D.~Antoff, G.~Blanford, V.~Bosetti,
		P.~Christensen, H.~McJeon, J.~Reilly, and P.~Sztorc}, {\em {Uncertainty in
			Climate Change: A Multimodel Comparison}}, NBER Working Paper No. 21637,
	October (2015), pp.~1--25.
	
	\bibitem{Guo02}
	{\sc W.~Guo}, {\em {Inference in smoothing spline analysis of variance}},
	Journal of the Royal Statistical Society Series B, 66 (2002), pp.~887--898.
	
	\bibitem{Hall2012}
	{\sc J.~W. Hall, R.~J. Lempert, K.~Keller, A.~Hackbarth, C.~Mijere, and D.~J.
		Mcinerney}, {\em {Robust Climate Policies Under Uncertainty : A Comparison of
			Robust Decision Making and Info-Gap Methods}}, Risk Analysis, 32 (2012),
	pp.~1657--1672.
	
	\bibitem{Hoef48}
	{\sc W.~Hoeffding}, {\em {A class of statistics with asymptotically normal
			distribution}}, Annals of Mathematical Statistics, 19 (1948), pp.~293--325.
	
	\bibitem{Homma1996}
	{\sc T.~Homma and A.~Saltelli}, {\em {Importance Measures in Global Sensitivity
			Analysis of Nonlinear Models}}, Reliability Engineering {\&} System Safety,
	52 (1996), pp.~1--17.
	
	\bibitem{Hooker2007}
	{\sc G.~Hooker}, {\em {Generalized Functional ANOVA Diagnostics for High
			Dimensional Functions of Dependent Variables}}, Journal of Computational and
	Graphical Statistics, 16 (2007), pp.~709--732.
	
	\bibitem{HuCaoHong12}
	{\sc Z.~Hu, J.~Cao, and L.~Hong}, {\em {Robust Simulation of Global Warming
			Policies Using the DICE Model}}, Management Science, 58 (2012),
	pp.~2190--2206.
	
	\bibitem{Huang98JMA}
	{\sc J.~Z. Huang}, {\em {Functional ANOVA Models for Generalized Regression}},
	Journal of Multivariate Analysis, 67 (1998), pp.~49--71.
	
	\bibitem{HuangAS98}
	{\sc J.~Z. Huang}, {\em {Projection Estimation in Multiple Regression with
			Application to Functional Anova Models}}, The Annals of Statistics, 26
	(1998), pp.~242--272.
	
	\bibitem{HuangetalAS99}
	{\sc J.~Z. Huang, C.~Kooperberg, C.~J. Stone, and Y.~K. Truong}, {\em
		{Functional ANOVA modeling for proportional hazards regression}}, The Annals
	of Statistics, 28 (2000), pp.~961--999.
	
	\bibitem{IshiHomm90}
	{\sc T.~Ishigami and T.~Homma}, {\em {An Importance Quantification Technique in
			Uncertainty Analysis for Computer Models}}, in ISUMA'90, First International
	Symposium on Uncertainty Modelling and Analysis, University of Maryland,
	1990.
	
	\bibitem{KaufSain10Bayesian}
	{\sc C.~G. Kaufman and S.~R. Sain}, {\em {Bayesian functional ANOVA modeling
			using Gaussian process prior distributions}}, Bayesian Analysis, 5 (2010),
	pp.~123--149.
	
	\bibitem{Li2012}
	{\sc G.~Li and H.~Rabitz}, {\em {General Formulation of HDMR Component
			Functions with Independent and Correlated Variables}}, Journal of
	Mathematical Chemistry, 50 (2012), pp.~99--130.
	
	\bibitem{LiRabi17}
	{\sc G.~Li and H.~Rabitz}, {\em {Relationship between Sensitivity Indices
			Defined by Variance- and Covariance-Based Methods}}, Reliability Engineering
	{\&} System Safety, 167 (2017), pp.~136--157,
	\href{http://dx.doi.org/10.1016/j.ress.2017.05.038}
	{doi:10.1016/j.ress.2017.05.038}.
	
	\bibitem{LinZhang06}
	{\sc Y.~Lin and H.~H. Zhang}, {\em {Component selection and smoothing in
			multivariate nonparametric regression}}, The Annals of Statistics, 34 (2006),
	pp.~2272--2297.
	
	\bibitem{Marinacci2005311}
	{\sc M.~Marinacci and L.~Montrucchio}, {\em {Ultramodular functions}},
	Mathematics of Operations Research, 30 (2005), pp.~311--332.
	
	\bibitem{Marinacci2008642}
	{\sc M.~Marinacci and L.~Montrucchio}, {\em {On concavity and
			supermodularity}}, Journal of Mathematical Analysis and Applications, 344
	(2008), pp.~642--654.
	
	\bibitem{McInLempKell11}
	{\sc D.~McInerney, R.~Lempert, and K.~Keller}, {\em {What are Robust Strategies
			in the Face of Uncertainty}}, Climatic Change, 91 (2011), pp.~29--41.
	
	\bibitem{McInKell08}
	{\sc D.~J. Mcinerney and K.~Keller}, {\em {Economically optimal risk reduction
			strategies in the face of uncertain climate thresholds}}, Climatic Change, 91
	(2008), pp.~29--41.
	
	\bibitem{MilgShan94}
	{\sc P.~Milgrom and C.~Shannon}, {\em {Monotone Comparative Statics}},
	Econometrica, 62 (1994), pp.~157--180.
	
	\bibitem{MillDietHeal12}
	{\sc A.~Millner, S.~Dietz, and G.~Heal}, {\em {Scientific Ambiguity and Climate
			Policy}}, Environmental and Resource Economics, 55 (2013), pp.~21--46.
	
	\bibitem{Grang14PNAS}
	{\sc M.~Morgan}, {\em {Use (and abuse) of expert elicitation in support of
			decision making for public policy}}, PNAS, 111 (2014), pp.~7176--7184.
	
	\bibitem{Nord08}
	{\sc W.~Nordhaus}, {\em {A Question of Balance: Weighing the Options on Global
			Warming Policies}}, Yale University Press, New Haven, NJ, USA, 2008.
	
	\bibitem{OaklOhag04}
	{\sc J.~Oakley and A.~O'Hagan}, {\em {Probabilistic Sensitivity Analysis of
			Complex Models: a Bayesian Approach}}, Journal of the Royal Statistical
	Society, Series B, 66 (2004), pp.~751--769.
	
	\bibitem{OHagan2006}
	{\sc A.~O'Hagan, C.~Buck, A.~Daneshkhah, J.~{Richard Eiser}, P.~Garthwaite,
		D.~Jenkinson, J.~Oakley, T.~Rakow, and I.~978-0-470-02999-2}, {\em {Uncertain
			Judgements: Eliciting Experts Probabilities}}, Wiley and Sons, UK, 2006.
	
	\bibitem{Ouchi2004}
	{\sc F.~Ouchi}, {\em {A Literature Review on the Use of Expert Opinion in
			Probabilistic Risk Assessment}}, World Bank Policy Research Working Paper
	3201,  (2004), pp.~1--17.
	
	\bibitem{Owen03}
	{\sc A.~B. Owen}, {\em {The Dimension Distribution and Quadrature Test
			Functions}}, Statistica Sinica, 13 (2003), pp.~1--17.
	
	\bibitem{Owen12b}
	{\sc A.~B. Owen}, {\em {Better Estimation of Small Sobol Sensitivity Indices}},
	ACM Transactions on Modeling and Computer Simulation, 23 (2013), p.~11.
	
	\bibitem{Owen13SIAM}
	{\sc A.~B. Owen}, {\em {Variance Components and Generalized Sobol Indices}},
	SIAM/ASA Journal on Uncertainty Quantification, 1 (2013), pp.~19--41.
	
	\bibitem{PaleConf16}
	{\sc L.~Paleari and R.~Confalonieri}, {\em {Sensitivity Analysis of a
			Sensitivity Analysis: We are Likely Overlooking the Impact of Distributional
			Assumptions}}, Ecological Modelling, 340 (2016), pp.~57--63.
	
	\bibitem{Plis12EMS}
	{\sc E.~Plischke}, {\em {How to Compute Variance-Based Sensitivity Indicators
			with Your Spreadsheet Software}}, Environmental Modelling {\&} Software, 35
	(2012), pp.~188--191.
	
	\bibitem{PlisBorg13}
	{\sc E.~Plischke and E.~Borgonovo}, {\em {What about totals? Alternative
			approaches to factor fixing}}, in Safety, Reliability and Risk Analysis:
	Beyond the Horizon - Proceedings of the European Safety and Reliability
	Conference, ESREL 2013, 2014, pp.~3339--3344.
	
	\bibitem{Rabitz1999}
	{\sc H.~Rabitz and O.~Alis}, {\em {General foundations of High-Dimensional
			Model Representations}}, J. Math. Chem., 25 (1999), pp.~197--233.
	
	\bibitem{Rahman20082091}
	{\sc S.~Rahman}, {\em {A polynomial dimensional decomposition for stochastic
			computing}}, International Journal for Numerical Methods in Engineering, 76
	(2008), pp.~2091--2116.
	
	\bibitem{Rahman2011a}
	{\sc S.~Rahman}, {\em {Global Sensitivity Analysis by Polynomial Dimensional
			Decomposition}}, Reliability Engineering {\&} System Safety, 96 (2011),
	pp.~825--837.
	
	\bibitem{Rahm14}
	{\sc S.~Rahman}, {\em {A Generalized ANOVA Dimensional Decomposition for
			Dependent Probability Measures}}, SIAM/ASA Journal on Uncertainty
	Quantification, 2 (2014), pp.~670--697.
	
	\bibitem{Rahman201127}
	{\sc S.~Rahman and A.~Chakraborty}, {\em {Stochastic multiscale fracture
			analysis of three-dimensional functionally graded composites}}, Engineering
	Fracture Mechanics, 78 (2011), pp.~27--46.
	
	\bibitem{RahmYada11IJUQ}
	{\sc S.~Rahman and V.~Yadav}, {\em {Orthogonal Polynomial Expansions for
			Solving Random Eigenvalue Problems}}, international Journal for Uncertainty
	Quantification, 1 (2011), pp.~163--187.
	
	\bibitem{Ratto2010}
	{\sc M.~Ratto and A.~Pagano}, {\em {Using Recursive Algorithms for the
			Efficient Identification of Smoothing Spline ANOVA Models}}, Advances in
	Statistical Analysis, 94 (2010), pp.~367--388.
	
	\bibitem{ratto_cpc_2007}
	{\sc M.~Ratto, A.~Pagano, and P.~Young}, {\em {State Dependent Parameter
			metamodelling and sensitivity analysis}}, Computer Physics Communications,
	177 (2007), pp.~863--876.
	
	\bibitem{Ren2016425}
	{\sc X.~Ren, V.~Yadav, and S.~Rahman}, {\em {Reliability-based design
			optimization by adaptive-sparse polynomial dimensional decomposition}},
	Structural and Multidisciplinary Optimization, 53 (2016), pp.~425--452.
	
	\bibitem{Saltelli_cpc_2002}
	{\sc A.~Saltelli}, {\em {Making Best Use of Model Valuations to Compute
			Sensitivity Indices}}, Computer Physics Communications, 145 (2002),
	pp.~280--297.
	
	\bibitem{Saltelli2010a}
	{\sc A.~Saltelli, P.~Annoni, I.~Azzini, F.~Campolongo, M.~Ratto, and
		S.~Tarantola}, {\em {Variance based sensitivity analysis of model output.
			{\{}D{\}}esign and estimator for the total sensitivity index}}, 181 (2010),
	pp.~259--270.
	
	\bibitem{Salt10TotalCPC}
	{\sc A.~Saltelli, P.~Annoni, I.~Azzini, F.~Campolongo, M.~Ratto, and
		S.~Tarantola}, {\em {Variance based sensitivity analysis of model output.
			Design and estimator for the total sensitivity index}}, Computer Physics
	Communications, 181 (2010), pp.~259--270.
	
	\bibitem{Saltelli2008}
	{\sc A.~Saltelli, M.~Ratto, T.~Andres, F.~Campolongo, J.~Cariboni, D.~Gatelli,
		M.~Saisana, and S.~Tarantola}, {\em {Global Sensitivity Analysis -- The
			Primer}}, Chichester, 2008.
	
	\bibitem{SaltIST15}
	{\sc A.~Saltelli, P.~Stark, W.~Becker, and P.~Stano}, {\em {Climate Models as
			Economic Guides: Scientific Challeng or Quixotic Quest?}}, Issues in Science
	and Technology,  (2015), pp.~79--84.
	
	\bibitem{saltelli_jasa_2002}
	{\sc A.~Saltelli and S.~Tarantola}, {\em {On the Relative Importance of Input
			Factors in Mathematical Models: Safety Assessment for Nuclear Waste
			Disposal}}, Journal of the American Statistical Association, 97 (2002),
	pp.~702--709.
	
	\bibitem{Saltelli2000}
	{\sc A.~Saltelli, S.~Tarantola, and F.~Campolongo}, {\em {Sensitivity Analysis
			as an Ingredient of Modelling}}, Statistical Science, 19 (2000),
	pp.~377--395.
	
	\bibitem{Saltelli1999}
	{\sc A.~Saltelli, S.~Tarantola, and K.~Chan}, {\em {A Quantitative, Model
			Independent Method for Global Sensitivity Analysis of Model Output}},
	Technometrics, 41 (1999), pp.~39--56.
	
	\bibitem{sobol_mmce_1993}
	{\sc I.~M. Sobol'}, {\em {Sensitivity analysis for non-linear mathematical
			models}}, Mathematical Modelling and Computational Experiment, 1 (1993),
	pp.~407--414.
	
	\bibitem{Sobo93}
	{\sc I.~M. Sobol'}, {\em {Sensitivity Estimates for Nonlinear Mathematical
			Models}}, Mathematical Modelling {\&} Computational Experiments, 1 (1993),
	pp.~407--414.
	
	\bibitem{StroOaki12JRSSC}
	{\sc M.~Strong, J.~E. Oakley, and J.~Chilcott}, {\em {Managing Structural
			Uncertainty in Health Economic Decision Models: a Discrepancy Approach}},
	Journal of the Royal Statistical Society, Series C, 61 (2012), pp.~25--45.
	
	\bibitem{Sudret2008}
	{\sc B.~Sudret}, {\em {Global sensitivity analysis using polynomial chaos
			expansion}}, 93 (2008), pp.~964--979.
	
	\bibitem{Tarantola2006}
	{\sc S.~Tarantola, D.~Gatelli, and T.~A. Mara}, {\em {Random balance designs
			for the estimation of first order global sensitivity indices}}, 91 (2006),
	pp.~717--727.
	
	\bibitem{Wahba78}
	{\sc G.~Wahba}, {\em {Improper Priors, Spline Smoothing and the Problem of
			Guarding Against Model Errors in Regression}}, Journal of the Royal
	Statistical Society Series B, 40 (1978), pp.~364--372.
	
	\bibitem{Wald47}
	{\sc A.~Wald}, {\em {Foundations of a General Theory of Sequential Decision
			Functions}}, Econometrica, 15 (1947), pp.~279--313.
	
	\bibitem{Wang06}
	{\sc X.~Wang}, {\em {On the Effects of Dimension Reduction Techniques on Some
			High-Dimensional Problems in Finance}}, Operations Research, 54 (2006),
	pp.~1063--1078.
	
	\bibitem{Wiener38}
	{\sc N.~Wiener}, {\em {The Homogeneous Chaos}}, American Journal of
	Mathematics, 60 (1938), pp.~897--936.
	
	\bibitem{XiuKarn02}
	{\sc D.~Xiu and G.~Karniadakis}, {\em {The Wiener-Askey Polynomial Chaos for
			Stochastic Differential Equations}}, SIAM Journal on Scientific Computing, 24
	(2002), pp.~619--644.
	
\end{thebibliography}

\end{document}